\documentclass[10pt,journal,compsoc]{IEEEtran}
\ifCLASSOPTIONcompsoc
  \usepackage[nocompress]{cite}
\else
  \usepackage{cite}
\fi
\ifCLASSINFOpdf
\else
\fi

\usepackage{cite}

\usepackage{setspace}
\usepackage{amsmath, amsthm, amssymb}
\usepackage{algorithm}
\usepackage[noend]{algpseudocode}
\makeatletter
\def\BState{\State\hskip-\ALG@thistlm}
\makeatother

\usepackage{float}
\usepackage{subfig}
\usepackage{graphicx}
\usepackage{caption}
\usepackage{mathrsfs}

\newtheorem{theorem}{\textbf{Theorem}}

\newtheorem{lemma}{\textbf{Lemma}}

\newtheorem{corollary}{\textbf{Corollary}}

\theoremstyle{definition}
\newtheorem{definition}{\textbf{Definition}}
\newtheorem{property}{\textbf{Property}}

\newtheorem{myrule}{Rule}
\newtheorem{problem}{Problem}
\newtheorem{example}{Example}

\DeclareMathOperator*{\argmax}{arg\,max}
\begin{document}
%
\title{Adaptive Influence Maximization in Dynamic Social Networks}
%
%
%
%

\author{Guangmo Tong,~\IEEEmembership{Student Member,~IEEE,}
		Weili Wu,
		Shaojie Tang,
        Ding-Zhu Du,~\IEEEmembership{Member,~IEEE,}
        
\IEEEcompsocitemizethanks{\IEEEcompsocthanksitem G. Tong and D.-Z. Du are with the Department of Computer Science Erik
Jonsson School of Engineering and Computer Science The University
of Texas at Dallas 800 W. Campbell Road; MS EC31 Richardson, TX
75080 U.S.A.\protect\\
E-mail: \{guangmo.tong, dzdu\}@utdallas.edu
}
\thanks{}}

%
%

\markboth{Journal of \LaTeX\ Class Files,~Vol.~13, No.~9, September~2014}%
{Shell \MakeLowercase{\textit{et al.}}: Bare Demo of IEEEtran.cls for Computer Society Journals}
%



\IEEEtitleabstractindextext{%
\begin{abstract}
For the purpose of propagating information and ideas through a social network, a seeding strategy aims to find a small set of seed users that are able to maximize the spread of the influence, which is termed as influence maximization problem. Despite a large number of works have studied this problem, the existing seeding strategies are limited to the static social networks. In fact, due to the high speed data transmission and the large population of participants, the diffusion processes in real-world social networks have many aspects of uncertainness. Unfortunately, as shown in the experiments, in such cases the state-of-art seeding strategies are pessimistic as they fails to trace the dynamic changes in a social network. In this paper, we study the strategies selecting seed users in an adaptive manner. We first formally model the Dynamic Independent Cascade model and introduce the concept of adaptive seeding strategy. Then based on the proposed model, we show that a simple greedy adaptive seeding strategy finds an effective solution with a provable performance guarantee. Besides the greedy algorithm an efficient heuristic algorithm is provided in order to meet practical requirements. Extensive experiments have been performed on both the real-world networks and synthetic power-law networks. The results herein demonstrate the superiority of the adaptive seeding strategies to other standard methods.
\end{abstract}

}

\maketitle

\IEEEdisplaynontitleabstractindextext

%
\IEEEpeerreviewmaketitle

\IEEEraisesectionheading{\section{Introduction}\label{sec:introduction}}

%
%
%
%
\IEEEPARstart{W}{ith} the advance of information science in the last two decades, social networks are becoming important dissemination platforms as they allow efficient interchange of ideas and information. The influence diffusion process in social networks has been studied in many domains e.g. epidemiology, social median and economics. It has been shown that the investigation into the influence diffusion are of great use in many aspects such as designing marketing strategy \cite{mahajan1990new,goldenberg2001using}, analyzing human behavior \cite{bond201261} and rumor blocking \cite{fan2013least}. In order to formulate the diffusion process, a number of models have been studied during the last decade. Two basic operational models, linear threshold (LT) model and independent cascade (IC) model, are proposed by Kempe \textit{et al.} \cite{kempe2003maximizing}. In the Linear Threshold Model, a user will adopt a new idea if the influence from its neighbors has reached a certain threshold, while in the Independent Cascade Model an adopter has a certain probability to convince each of its neighbors. Based on those two models various models have been developed and studied.

In the topic of influence diffusion, an important issue is that how to propagate information through a social network effectively and efficiently.  As an example, in order to advertise new products, a company would like to offer free samples to a set of initial users who will potentially introduce the new product to their friends. Due to the expense issue, only a limited number of samples are available and thus we have a budget of the seed users. A natural problem is that how to select a good set of seed users that is able to maximize the number of customers who finally adopt the target product. This problem is named as influence maximization problem first proposed in \cite{domingos2001mining} in literature. 

A large body of related works have been performed concerning the influence maximization problem, but the state-of-art technique may not deal with many real cases in effect. A drawback of the existing diffusion models is that they fail to take account of some uncertain natures of a real-world social network. Such uncertainness can be viewed from the following three aspects. In a real-world social network, the seed users are not assured to be successfully activated. In the example of selling a new product, the advertising would be stuck if the free samples do not satisfy the initial users. Second, the information is not guaranteed to be delivered from one user to the other and thus the diffusion itself is a probabilistic process. Furthermore, the topology of a social network is not always static in real cases due to the frequent variation of the degree of the relationship between users. In the sense of an online social network, such as Facebook, Twitter or Flicker, topology changes are incurred by the increasing number of the common friends between a pair of users. In this paper, we study the influence maximization problem in the social networks with the above characteristics. By extending the classic IC model, we herein develop the Dynamic Independent Cascade (DIC) model which is able to capture the dynamic aspects of real social networks. In the classic IC model a seed node is guaranteed to be activated after selected and the relationship between two users is simply represented by a fixed probability, while the seed nodes in our DIC model could fail to be activated with a certain probability and the propagation probability between two users follows a certain distribution which reflects the change of topology of a social network. 

Based on the DIC model, we further consider how to design a seeding strategy to find effective seed nodes. For the classic IC model, Kempe \textit{et al.} \cite{kempe2003maximizing} propose a simple greedy algorithm with an approximation ratio of $(1-1/e)$ and Chen \textit{et al.} \cite{chen2010scalable} present an efficient heuristic seeding approach to handle large-scale social networks. The existing approaches always make seeding selection in a static manner (i.e., determining a seed set before the process of spread), which renders them inapplicable to the DIC model. As mentioned earlier the seed users in the DIC model are not guaranteed to be activated. In this setting, an arising issue is that we can seed a user for more than one time if it is not successfully activated in the past rounds. One can see that it is worthy to take more effort to activate a powerful user as he or she may generate considerable influence to a social network. However, a static seeding algorithm cannot take such a case into account. Besides, to determine a seed set, the prior algorithms require the propagation probability between users, but in the DIC model such a probability is a random variable and we can only expect a distribution over it. Admittedly, we could take advantage of its expected value and then apply the prior approach. But such a method would be pessimistic as it fails to trace the dynamic topology of a real-world social network. In this paper, we first provide a simple adaptive seeding strategy that is able to handle the dynamic aspects of real-world social networks, and then design a heuristic algorithm for better scalability.

\subsection{Related Work and Technique}
\label{subsec:relate}
Domingos \textit{et al.} \cite{domingos2001mining} are among the first who study the influential nodes in viral marketing. In the seminal work \cite{kempe2003maximizing}, Kempe \textit{et al.} formulate the influence maximization problem from the view of combinatorial optimization, and provide a greedy algorithm with an approximation ratio of $(1-1/e)$. Efficient heuristic influence maximization algorithms have been studied in many works \cite{chen2010scalable}, \cite{chen2009efficient} and \cite{chen2009approximability}. Long \textit{et at.} \cite{long2011minimizing} further study this problem from the perspective of minimization. Du \textit{et al.} \cite{du2013scalable} and Rodriguez \textit{et al.} \cite{rodriguez2012influence} propose the continuous diffusion model and study the influence maximization problem in this setting. All the above works aim to determine an effective seed set before the diffusion process and focus on the network with a static topology.  

In order to learn a provable performance guarantee, submodular functions play an important role in the prior works. Kempe \textit{et al.} \cite{kempe2003maximizing} show that the expected number of active nodes is a monotonically increasing submodular function over the seed set, and therefore, by the celebrated result in \cite{nemhauser1978analysis}, a simple greedy algorithm yields an $(1-1/e)$ approximation. However, as shown later in Sec. \ref{sec:greedy}, such a technique cannot be directly applied to the adaptive seeding problem. On the one hand the seed nodes are unknown before the diffusion process as they are adaptively selected; on the other hand the value of the objective function over a certain seed set cannot be explicitly observed.

Adaptive seeding strategy is a stochastic optimization framework and a natural extend to original seeding approach in \cite{kempe2003maximizing}. Part of the analysis in this paper is based on the stochastic submodular maximization. Asadpour \textit{et al.} \cite{asadpour2008stochastic} present the analysis of the stochastic submodular maximization problem where the objective function is defined on the power set of a set of independent random variables. Golovin \textit{et al.} \cite{golovin2010adaptive} further study this problem with the concept of adaptive submodularity. Although the above works are only applicable to special cases of the adaptive influence maximization problem, they provide a clue that the greedy algorithm in its adaptive version is still able to achieve a provable performance guarantee. 
In a recent work,  Seeman \textit{et al.} \cite{seeman2013adaptive} consider the adaptive approach to a variant influence maximization problem where the seed nodes are constrained in a certain set and the influence can spread for only one round, and thus has a different setting from that of this paper.

\subsection{Contribution}
\label{subsec:contribution}
The contributions of this paper are summarized as follows. We propose the DIC model that is able to capture the dynamic aspects of real-world social networks. In order to provide a formal description of an adaptive seeding strategy we introduce the concept of seeding pattern. The main contribution of this paper is an adaptive hill-climbing strategy with a provable performance guarantee in the DIC model. We further design an efficient heuristic adaptive seeding strategy by narrowing the candidate seed sets before the seeding process. The conducted experiments demonstrate the superiority of the proposed adaptive seeding strategies to the original seeding approaches in dynamic social networks.
 
The rest of the paper is organized as follow. The proposed DIC model and the adaptive seeding strategy are formulated in Sec. \ref{sec:model}. The analysis of the greedy adaptive strategy is shown in Sec. \ref{sec:greedy} and the heuristic strategy is proposed in Sec. \ref{sec:heuristic}. In Sec. \ref{sec:exp}, we show the experimental results. Sec. \ref{sec:conclusion} concludes.

\section{Problem setting}
\label{sec:model}
\subsection{DIC Model}
\label{subsec:dic_model}
A social network is modeled as a directed graph where nodes and edges denote the individuals and social ties, respectively. In order to spread an idea or advertising a new product in a social network, some seed nodes are chosen to be activated (e.g., by giving payments or offering free samples) to trigger the spread of influence. Following the notations in \cite{kempe2003maximizing} we speak of each node as being either \textit{active} or \textit{inactive}. A node can be activated either by its neighbor or as a seed node.

In the DIC model, associated with each node $u$ there is a random variable $X_u$ following a Bernoulli distribution $f_u$, where $X_u=1$ indicates node $u$ is successfully activated as a seed node.  For the relationship between nodes, an active node $u$ has one chance to activate its inactive neighbor $v$ via edge $(u,v)$ with a probability of $X_{(u,v)}$ which is a random variable. With the activated seed nodes diffusion process goes round by round. Without the loss of generality, for each edge $e$, we assume $X_{e}$ follows a certain discrete distribution $f_{e}$ with a domain $D_{e}$, and let $d_e^i \in [0,1]$ be the $i^{th}$ value in $D_{e}$. In this paper, we do not enforce any specific distribution of $X_e$ \footnote{We may assume an exponential distribution as a social network always exhibits a power-law pattern where the influential users are rare \cite{clauset2009power}.}. In the DIC model, for an edge $e=(u,v)$, the value of $X_e$ remains unknown until one of the neighbors of $u$ is active. This is because in practice an industry institute may only trace the interested influence and the real-time state of the rest of the network is unavailable. We denote an instance of DIC network by $G=(V,E,F_V,F_E)$, where $F_V=\{f_u |u \in V\}$ and $F_E=\{f_e | e \in E\}$ are the sets of the distributions of $X_u$ and $X_e$, respectively. Let $N$ be the number of the nodes in $V$. Due to the expense of activating seed nodes, we have a budget $B (B \leq N)$ for the seed set. The notations that are frequently used later in this paper are listed in Table \ref{table:symbol} and the rest of the notations in Table \ref{table:symbol} will be introduced later.  

\begin{table}[t]

\centering
{\begin{tabular}{ |p{2.1cm} || p{5.8cm} |}
\hline 
Symbol& Definition \\
\hline 
$G$& Instance of DIC network.    \\ 
\hline
$G_1$& Example DIC network in Example \ref{exp:intro}.    \\ 
\hline
$B$& Budget of seed set.    \\ 
\hline
$D_e$& Domain of the propagation probability of edge $e$.  \\ 
\hline
$d_e^{i}$& The $i^{th}$ value in $D_e$.  \\ 
\hline
$\text{Prob}[X_u=1]$ & The probability that $X_u$ can be activated as a seed node when selected. \\
\hline
$A$&  Seed pattern.\\ 
\hline 
$A_0$&  Special seed pattern define in Def. \ref{def:B^0}.\\ 
\hline 
$A^{*}$&  Special seed pattern define in Def. \ref{def:B^*}.\\ 
\hline 
$S_A^{G}$&  Seeding strategy of pattern $A$ on $G$\\ 
\hline
$OPT_A^{G}$&  Optimal seeding strategy of pattern $A$ on $G$\\ 
\hline
$c$-$G$&  Auxiliary graph of network $G$\\ 
\hline 
$x$&  Full realization\\ 
\hline 
$y$&  Partial realization\\ 
\hline 
$\epsilon$&  Empty realization\\ 
\hline 
\end{tabular}}
\caption{Notations. }
\label{table:symbol}
\end{table}

\subsection{Adaptive Seeding Strategy}
\label{subsec:adaptive}
Basically, to design an adaptive seeding strategy we consider two problems: (1) how many budgets should we use in each seeding step and (2) which nodes to select.  We employ the following concepts to formulate those problems.

Assuming that the seed nodes are only selected between two spread rounds, we denote the seeding step between round $i-1$ and round $i$ as the $i^{th}$ seeding step, and the first seeding step is executed before the process of spread. We assume that we need one round to activate the seed nodes selected in each seeding step. In this paper, we consistently use ``step" for seeding process and ``round" for diffusion process.

\begin{definition}
\label{def:pattern}
A \textit{seeding pattern} $A=(a_1,...,a_N)$ is a sequence of non-negative integers, implying that we seed $a_i$ nodes in the $i^{th}$ seeding step. We will later show that we have at most $N$ seeding steps. Due to the budget constraint, $\sum a_i \leq B$. Note that it reduces to the non-adaptive seeding if $A=(B)$. Corresponding to a seeding pattern $A=(a_1,...,a_N)$, a \textit{seeding strategy} $S_A=(s_1,...,s_N)$ of $A$ is a sequence of node-sets where $|s_i|=a_i$ and $s_i$ is the node-set seeded in the $i^{th}$seeding step. That $a_i=0$ implies that we do not seed any node in the $i^{th}$seeding step and thus $s_i=\emptyset$. 
\end{definition}

In the above setting, both the seeding pattern and seeding strategy can be adaptively constructed, i.e., $a_i$ and $s_i$ may depend on the outcomes of the past rounds. For a specific DIC network $G$, we use $S_A^{G}$ to denote a seeding strategy of pattern $A$ on $G$. Since DIC model is a probabilistic model, the objective function herein is the expected number of the final active nodes when there is no node can be further activated and no budget left. We denote the expected number of active nodes in $G$ under a seeding strategy $S_A^{G}$ by $E[S_A^{G}]$. 

\begin{definition}
\label{def:null_round}
Given a strategy $S_A^{G}$ on a DIC network $G$, if $s_i=\emptyset$ but there does not exist any edge $(u,v)$ such that $u$ is activated, either by its neighbors or as a seed node, in the $(i-1)^{th}$ round, we say that $S_A^{G}$ waits for a \textit{null round}. It can be easily seen that waiting for a null round has no impact on the process of spread or the effect of the  strategy. Unless otherwise stated, we assume that any strategy will not wait for one or more null rounds. Therefore, we have at most $N$ seeding steps and $s_1 \neq \emptyset$ for any strategy $S_A^{G}=(s_1,...,s_N)$. For the convenience of analysis, we require that any strategy $S_A^G$ will not select an active node as a seed node.
\end{definition}

Two natural patterns $A_0$ and $A^{*}$ are defined as follows.

\begin{definition}{}
\label{def:B^0}
Let $A_0=\{a_1,...,a_N\}$ where $a_i=1$ for $1 \leq i \leq B$ and $a_i=0$ for $i > B$. Informally, under pattern $A_0$ we successively seed one node in each step until the budget is used up. 
\end{definition}

\begin{definition}{}
\label{def:B^*}
Another pattern $A^{*}$ is adaptively constructed as follows. In pattern $A^{*}$, we seed one node at a time and wait until no node can be further activated before seeding the next node. Thus, we seed one node in the first step and the rest of seeding pattern will be constructed adaptively.
\end{definition}

Note that given a pattern $A$ there exists many strategies of $A$. We use $OPT_A^{G}$ to denote the optimal adaptive strategy of pattern $A$ on a given DIC network $G$ with respect to the expected number of active nodes. 

The core problem considered in this paper is defined as follows.

\begin{problem}{Adaptive Influence Maximization (AIM).}
Under the budget constraint, for any DIC network $G$, find a pattern $A$ and a strategy $S_A^{G}$ of $A$ on $G$ such that $E[S_A^G]$ is maximized.
\end{problem}

\subsection{An Example}
We employ the following example to illustrate the DIC model and the concept of seeding pattern.
\begin{example}
\label{exp:intro}
Consider an example DIC network $G_1=(V,E,F_V,F_E)$ with six nodes and five edges, as shown in Fig. \ref{fig:example}, where $f_{v}(1)=0.5$ for each $v \in V$, and $D_{e}=\{0.4,0.8\}$ with $f_{e}(0.4)=0.8$ for each $e \in E$. In this example, each node can be activated with a probability of 0.5 when selected as a seed node, and the propagation probability between two connected nodes could be $0.4$ or $0.8$ with probabilities $0.8$ and $0.2$, respectively. We set the budget $B$ to be three. Suppose a certain seeding strategy $S_{A_1}^{G_1}$ produces a sequence of seed sets as $(\{v_3\},\{v_3\},\emptyset,\{v_1\})$ of pattern $A_1=(1,1,0,1)$. In this concrete seeding process, $S_{A_1}^{G_1}$ seeds $v_3$ twice respectively in step 1 and 2, which implies it fails to activate $v_3$ in the first time. Such a strategy may depend on the outcomes of the past rounds or the propagation probability observed in each step. 
\end{example}
\begin{figure}[t]
\begin{center}
\includegraphics[width=3.5in]{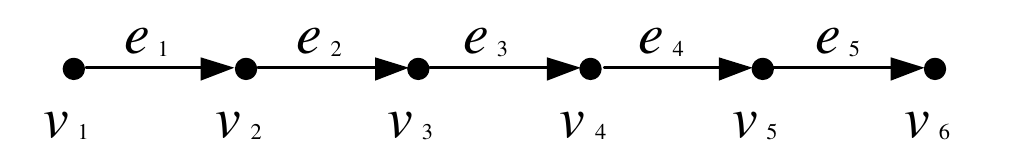} 
\end{center} 
\caption{Example network $G_1$.}
\label{fig:example}
\vspace{-4mm}
\end{figure}

\section{Greedy Algorithm}
\label{sec:greedy}
In this section, we show the main result of this paper. The seed selection rule of the greedy algorithm is shown as follows.
\begin{myrule}
\label{rule:greedy}
In each seeding step, we select the node that is able to maximize the marginal profit conditioned on the observed events.
\end{myrule}

Note that in each step we can observe the followings: (1) the outcome of the past rounds; (2) the propagation probabilities between the active nodes and their neighbors. We can see that Rule \ref{rule:greedy} can be applied to any pattern. For a pattern $A$ and a DIC network $G$, we use $\overline{S}_{A}^{G}$ to denote the seeding strategy following Rule \ref{rule:greedy}. Our analysis consists of three steps. First, we propose a transformation approach which finds an explicit expression of the expected number of the active nodes. Then, we prove that $A^{*}$ is the optimal pattern for any DIC network $G$, i.e., for any pattern $A^{'}$, $E[OPT_{A^{*}}^G] \geq E[OPT_{A^{'}}^G]$ . Finally, we show that $\overline{S}_{A^{*}}^{G}$ is a $(1-1/e)$-approximation under pattern $A^{*}$, i.e., 
\begin{equation}
E[\overline{S}_{A^{*}}^{G}] \geq  (1-1/e) \cdot E[OPT_{A^{*}}^G].
\end{equation}
\subsection{Transformation}
In the classic IC model, a concrete network is a graph where each edge $(u,v)$ is specified to be either \textit{live} or \textit{not live}. If edge $(u,v)$ is live then it means $u$ could successfully activate $v$. Informally speaking, all the uncertainties are determinate in a concrete network. In a concrete network, the active nodes are those which are connected to a seed node via a path of live edges, and the number of the active nodes in a concrete network is a submodular function over the seed set \cite{kempe2003maximizing}. Unfortunately, this approach cannot be directly applied to the analysis of our DIC model because several cases in the DIC model cannot be represented by a graph with a structure identical to that of the original DIC network. For example, how to represent the case that we seed a node more than once, and how to depict the feature that each propagation probability follows a distribution instead of being a single value? To address such scenarios, we transfer the original network to an auxiliary graph where the active nodes can be explicit observed given a seed set.

Given a DIC network $G=(V,E,F_V,F_E)$ where $V=\{v_1,...,v_N\}$, we construct an auxiliary graph $c$-$G=(V_c,E_c)$, as follows. $V_c$ consists  of $N \cdot B+N$ nodes and is partitioned into $N+1$ subsets denoted by $V_c^i$ ($0 \leq i \leq N$), where $|V_c^0|=N$ and $|V_c^i|=B$ ($i >0$). Let $V_c^0=\{v_{0,1},...,v_{0,N}\}$ and $V_c^i=\{v_{i,1},...,v_{i,B}\}$ ($i>0$). Nodes in $V_c^0$ are corresponding to the nodes in $G$ and nodes in $V_c^i$ ($i>0$) are used to represent the multiple seedings on $v_i$ in $G$. $E_c$ consists of two parts $E_c^1$ and $E_c^2$, defined as follows. For $i>0$ and $1 \leq j \leq B$, we have an edge $(v_{i,j},v_{0,i})$ for each pair of $v_{i,j}$ and $v_{0,i}$, and for each pair of nodes $v_{0,i}$ and $v_{0,j}$ in $V_0$ ($1 \leq i \neq j \leq N$), we have $|D_{(v_i,v_j)}|$ edges denoted by $e_{i,j}^k$ ($1\leq k \leq |D_{(v_i,v_j)}|$) connecting $v_{0,i}$ to $v_{0,j}$. Let $E_c^1$ be the set of edges between $V_c^0$ and $V_c^i$ ($i>0$) and $E_c^2$ be the set of edges within $V_0^i$. Recall that  $D_{(v_i,v_j)}$ is the domain of $f_{(v_i,v_j)}$ which is the distribution of the propagation probability of edge $(v_i,v_j)$ in $G$.

The auxiliary graph $c$-$G_1$ of $G_1$ in Example \ref{exp:intro} is shown as Fig. \ref{fig:tran_example}. Further explanations are presented in the caption. 

\begin{figure}[t]
\begin{center}
\includegraphics[width=3.5in]{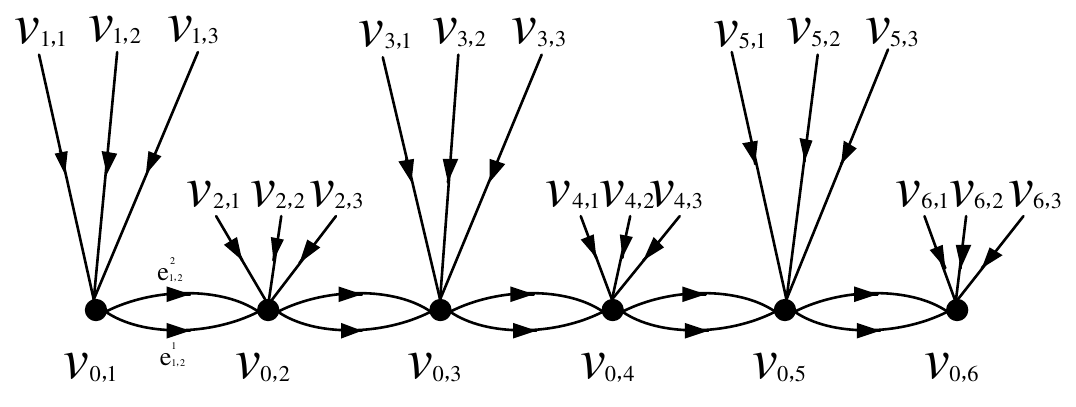} 
\end{center} 
\caption{\small Auxiliary graph $c$-$G_1$. In Example \ref{exp:intro}, we have a budget of three and the propagation probability of each edge in $G_1$ follows a distribution on a domain of two values. Therefore, we have three nodes $V_{1,1}$, $V_{1,2}$ and $V_{1,3}$ connected to $V_{0,1}$, and two edges $e_{1,2}^1$ and $e_{1,2}^2$ connecting $V_{0,1}$ and $V_{0,2}$.}
\label{fig:tran_example}
\vspace{-4mm}
\end{figure}

Now we show that given a seeding strategy how to observe the active nodes via $c$-$G$. Following the notations in \cite{asadpour2008stochastic}, we introduce the states of edges and the concept of realization.
\begin{definition}
A \textit{full realization} (\textit{f-realization}) $x$ of $c$-$G$ is a mapping from edges in $c$-$G$ to some states, where each edge in $E_c^{1}$ is mapped to \{\textit{live}, \textit{not live}\} and each edge in $E_c^2$ is mapped to \{\textit{selected-live}, \textit{selected-not live}, \textit{not selected}\}. In an f-realization, only one edge from $v_{0,i}$ to $v_{0,j}$ can be mapped to \textit{selected-live} or \textit{selected-not live}. 
\end{definition}

\begin{definition}
A \textit{partial realization} (\textit{p-realization}) $y$ of $c$-$G$ is a mapping from edges to states, where each edge in $E_c^{1}$ is mapped to \{\textit{live}, \textit{not live}, \textit{undetermined}\}, and each edge in $E_c^2$ is mapped to \{\textit{selected live}, \textit{selected-not live}, \textit{not selected}, \textit{selected-undetermined}, \textit{undetermined}\}. In a p-realization, if one edge from  $v_{0,i}$ to $v_{0,j}$ is \textit{undetermined} then all the edges from $v_{0,i}$ to $v_{0,j}$ must be \textit{undetermined}; if one edge from  $v_{0,i}$ to $v_{0,j}$ is either \textit{selected-live}, \textit{selected-live} or \textit{selected-undetermined}, then others edges from  $v_{0,i}$ to $v_{0,j}$ must be \textit{not selected}.
\end{definition}

\begin{table*}[t]
{\begin{tabular}{ |p{2.5cm} || p{13cm} |}

\hline 
\textit{live} &  $v_i$ in $G$ is successfully activated when selected as a seed node in the $j^{th}$ time. \\ 
\hline
\textit{not live} &  $v_i$ in $G$ fails to be activated when selected as a seed node in the $j^{th}$ time. \\ 
\hline
\textit{undetermined} &  The result of the $j^{th}$ seeding on $v_i$ is unknown. \\ 
\hline 
\end{tabular}}
\caption{States of edge $(v_{i,j},v_{0,i})$ in $E_c^1$, for $i>0$ and $1 \leq j \leq B$. }
\label{table:state1}
\end{table*}

\begin{table*}[t]

{\begin{tabular}{ | p{2.5cm} || p{13cm} |}

\hline 
\textit{selected-live} & The propagation probability between $v_i$ and $v_j$ is $d_{(v_i,v_j)}^k$ and $v_i$ activates $v_j$. \\ 
\hline

\textit{selected-not live} & The propagation probability between $v_i$ and $v_j$ is $d_{(v_i,v_j)}^k$ and $v_i$ fails to activate $v_j$. \\ 
\hline
\textit{selected-undetermined} & The propagation probability between $v_i$ and $v_j$ is $d_{(v_i,v_j)}^k$ and the result of the activation from $v_i$ to $v_j$ is unknown.\\
\hline
\textit{not selected} & The propagation probability between $v_i$ and $v_j$ is not $d_{(v_i,v_j)}^k$. \\
\hline
\textit{undetermined} & The propagation probability between $v_i$ and $v_j$ is unknown\\
\hline
\end{tabular}}

\caption{States of $e_{i,j}^k$ in $E_c^2$, for $1 \leq i \leq N$, $1 \leq j \leq N$ and $1 \leq k \leq |D_{(v_i,v_j)}|$.}
\label{table:state2}
\end{table*}

\begin{figure*}[t]
\begin{center}
\includegraphics[width=5in]{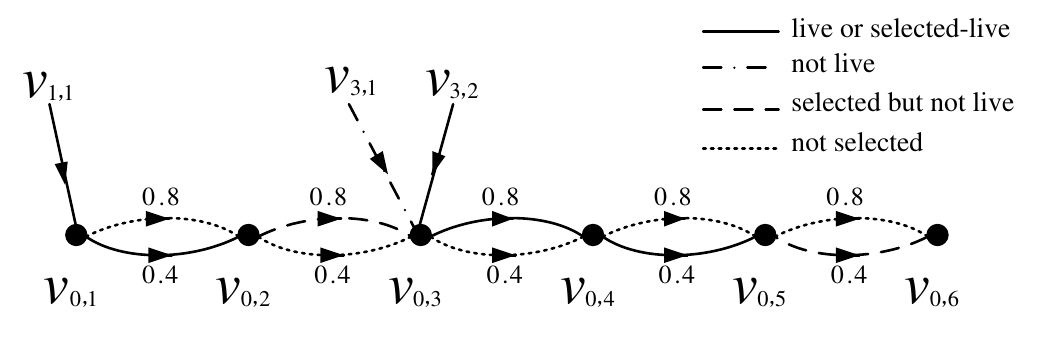} 
\end{center} 
\caption{An example f-realization $x_1$ of $c$-$G_1$. The number aligned with an edge is the propagation probability it stands for. In this concrete case, the seed nodes are $v_1$ and $v_3$, and the active nodes in $G$ are $v_1$, $v_3$, $v_4$ and $v_5$}
\label{fig:assign_example}
\end{figure*}

The explanations of the states are listed in Tables \ref{table:state1} and \ref{table:state2}. Each edge together with its state in $c$-$G$ corresponds to an event in the diffusion process of the original network $G$. We can see that an f-realization is a determinate case of the diffusion process and a p-realization is an intermediate state where the events are partially determined. For a seeding strategy $S_A^{G}$, the seed nodes selected by $S_A^{G}$ are determined only if an f-realization is specified. We use $S_A^{G^{x}}$ to denote the sequence of seed sets selected by $S_A^{G}$ under f-realization $x$.

For an f-realization $x$ and a p-realization $y$, let Prob[$x$] (resp. Prob[$y$]) be the probability with which $x$ (resp. $y$) happens and Prob[$x|y$] be the probability that $x$ happens conditioned on $y$.

\begin{definition}
An f-realization $x$ is \textit{compatible} to a p-realization $y$ if $x$ can be obtained from $y$ by changing the states of some edges in $y$ from \{\textit{undetermined}, \textit{selected-undetermined}\} into \{\textit{selected-live}, \textit{selected-not live}, \textit{not selected}\}. 
\end{definition}

Informally, $x$ is compatible to $y$ implies $x$ is a possible successive state of $y$ in the diffusion process. Similarly, we have the compatibility relationship between two p-realizations. Let $\epsilon$ be the empty realization where all the edges are in the \textit{undetermined} state. For a DIC network $G$, we denote the set of the f-realizations compatible to a p-realization $y$ by $C^G(y)$. 

For each strategy $S_A^{G}=(s_1,...,s_N)$ on $G=(V,E,F_V,F_E)$,  we have a corresponding seed set $V^{'} \subseteq \bigcup_{i>0}V_c^i$ in $c$-$G$,  constructed as follows. If $v_i$ in $G$ is selected by $S_A^{G}$ for $k$ times, then we add $v_{i,1}$,..., $v_{i,k}$ in $c$-$G$ to $V^{'}$. By this setting, given an f-realization $x$ of $c$-$G$, the number of active nodes under $S_A^{G}$ in $G$ is the number of the nodes in $V_c^0$ that are connected to a node in $V^{'}$ via \textit{live} edges in $c$-$G$. In the sense of Example. \ref{exp:intro}, an example f-realization $x_1$ with strategy $(\{v_3\},\{v_3\},\emptyset,\{v_1\})$ is illustrate in Fig. \ref{fig:assign_example}.

For an f-realization $x$, let $Node(S_A^{G^x})$ be the union of the corresponding seed sets produced by $S_A^{G^x}$ in $c$-$G$ in $x$. For a node-set $V^{'} \subseteq \bigcup_{i>0}V_c^i$, let $N_x^G(V^{'})$ be the number of active nodes in $x$ with seed set $V^{'}$. Therefore,
\begin{equation}
\label{eq:e_def}
E[S_A^G]=\sum_{x \in C^G(\epsilon)} \text{Prob}[x] \cdot N_x^G(Node(S_A^{G^x}))
\end{equation} 

$N_x^G(.)$ has the following important properties.

\begin{property}
\label{property:nondecreasing}
If $V_1 \subseteq V_2$, then $N_x^G(V_1) \leq N_x^G(V_2)$.  
\end{property}

\begin{property}
\label{property:submodular}
For two node-subsets $V_1$ and $V_2$ of $\bigcup_{i>0}V_i$, and a node $v^{'} \in \bigcup_{i>0}V_i$, where $V_1 \subseteq V_2$, $v^{'} \notin V_2$, we have 
\begin{eqnarray*}
N_x^G(V_2 \cup \{v^{'}\})- N_x^G(V_2 ) \leq N_x^G(V_1 \cup \{v^{'}\})-N_x^G(V_1 ). 
\end{eqnarray*}
\begin{proof}
This proof is similar to that of Theorem 2.2 in \cite{kempe2003maximizing}. The only difference is that, in our case, the seed nodes and active nodes are constrained in $\bigcup_{i>0}V_i$ and $V_0$, respectively.

\end{proof}
\end{property}




\begin{table*}[t]

\centering
{\begin{tabular}{ |p{0.6cm}  |p{1.2cm} |p{2.6cm} |p{2.6cm} |p{1.2cm} | p{2.6cm} |p{2.6cm} |}
\hline 
Step & $S_{A_1}^{G_1}$ & Diffusion process under $S_{A_1}^{G_1}$& Outcomes under $S_{A_1}^{G_1}$ & $S_{A^{*}}^{G_x}$  & Diffusion process under $S_{A^{*}}^{G_x}$& Outcomes under $S_{A^{*}}^{G_x}$ \\ 
\hline 
1 	& seeds $v_3$ &    & $v_3$ fails to be activated 
	& seeds $v_3$ &    & $v_3$ fails to be activated\\ 
\hline
2 	& seeds $v_3$ &    & $v_3$ is activated 
	& seeds $v_3$ &    & $v_3$ is activated\\ 
\hline
3 	&   &  $v_3$ activates $v_4$ & $v_4$ is activated 
	&   &  $v_3$ activates $v_4$ & $v_4$ is activated\\ 
\hline 
4 	& seeds $v_1$ 	&  $v_4$ activates $v_5$ & $v_5$ is activated; \newline $v_1$ is activated
	&   			&  $v_4$ activates $v_5$ & $v_5$ is activated\\ 
\hline
5 	&   			&  $v_1$ activates $v_2$ ; \newline $v_5$ activates $v_6$& $v_2$ is activated; \newline $v_6$ fails to be activated
	&  		 	&  $v_5$ activates $v_6$ & $v_6$ fails to be activated\\ 
\hline 
6 	&   &   &  
	& seed $v_1$ &    & $v_1$ is activated\\ 
\hline
7 	&   &   &  
	&   &  $v_1$ activates $v_2$ & $v_2$ is activated\\ 
\hline
\end{tabular}}
\caption{Seeding processes of $S_{A^{'}}^{G_1^{x_1}}$ and $S_{A^{*}}^{G_1^{x_1}}$. }
\label{table:process}
\end{table*}

\subsection{Optimal Pattern}
\label{subsec:opt_pattern}
As introduced in Sec. \ref{subsec:dic_model}, a seeding pattern identifies how many budgets should we consume in each step. Now, we show that $A^{*}$ is the optimal pattern.

\begin{lemma}
\label{lemma:result_1}
For any DIC network $G$, suppose $A^{'}$ is an arbitrary seeding pattern and $S_{A^{'}}^G$ is a known seeding strategy of $A^{'}$ on $G$ . There exist a seeding strategy $S_{A^{*}}^G$ of $A^{*}$ on $G$ such that $E[S_{A^{*}}^G]=E[S_{A^{'}}^G]$. 
\end{lemma}

\begin{proof}
The main idea is to construct a strategy $S_{A^{*}}^G$ according to $S_{A^{'}}^G$ such that, in any f-realization $x$, $ N_x^G(Node(S_{A^{'}}^{G^x}))=N_x^G(Node(S_{A^{*}}^{G^x}))$. 

Let ${\overline{x}}$ be an arbitrary but unknown f-realization of $c$-$G$. Suppose $S_{A^{'}}^{G^{\overline{x}}}=(s_1,...,s_N)$ and $A^{'}=(a_1,...,a_N)$. Assume $s_i=\{v_{i,1},...,v_{i,a_i}\}$ where the nodes are randomly ordered. Note that $s_1$ is known before the process of spread and $s_i$ ($i > 1$) is unknown until step $i$ as it depends on the outcomes of the past rounds. Let $Q$ be the sequence of the nodes in $\cup s_i$, where the nodes are non-decreasingly ordered by the nodes index in $s_i$ according to the lexicographical order. Following pattern $A^{*}$, let $S_{A^{*}}^{G^{\overline{x}}}$ choose the node in $Q$ in order. For the example shown in Fig. \ref{fig:assign_example} with f-realization $x_1$, the seeding process of strategy $S_{A_1}^{G_1^{x_1}}$ and its corresponding strategy $S_{A^{*}}^{G_1^{x_1}}$ are shown in Table \ref{table:process}.

One can see that $S_{A^{*}}^G$ does nothing but choose the nodes that are chosen by $S_{A^{'}}^G$. Note that although $S_{A^{'}}^G$ is known to us, the seed nodes produced by $S_{A^{'}}^G$ are undetermined as they depends on ${\overline{x}}$. Suppose $S_{A^{*}}^{G^{\overline{x}}}$ selects $v_{i,j}$ in the $l^{th}$ step, and the p-realizations in step $i$ under $S_{A^{'}}^{G^{\overline{x}}}$ and that under $S_{A^{*}}^{G^{\overline{x}}}$ in step $l$  are $y_1$ and $y_2$, respectively. To guarantee the feasibility of the construction of $S_{A^{*}}^{G^{\overline{x}}}$, $y_2$ must be compatible to $y_1$, which means, in realization ${\overline{x}}$, the events happening by step $i$  under strategy $S_{A^{'}}^{G^{\overline{x}}}$ is a subset of that of happening by step $l$ under strategy $S_{A^{*}}^{G^{\overline{x}}}$. For otherwise, in step $l$, $S_{A^{*}}^G$ cannot determine which node $v_{i,j}$ is.

In fact, such feasibility can be guaranteed by pattern $A^{*}$. Let $\overline{v}_i$ be the $i^{th}$ node in $Q$. Suppose $S_{A^{'}}^{G_{\overline{x}}}$ and $S_{A^{*}}^{G^{\overline{x}}}$ seeds $\overline{v}_i$ in step $l_i^{'}$ and step $l_i^{*}$, respectively.
Let $y_i^{'}$ (resp. $y_i^{*}$) be the p-realization under $S_{A^{'}}^{G^{\overline{x}}}$ (resp. $S_{A^{*}}^{G^{\overline{x}}}$) by step $l_i^{'}$ (resp. $l_i^{*}$). We need to prove that $y_{i}^*$ is compatible to $y_i^{'}$, for any $i>1$. We prove it by induction. Clearly, $y_{1}^{*}$ is compatible to $y_{1}^{'}$ as $y_{1}^{*}=y_{1}^{'}=\epsilon$. Suppose $y_{i}^{*}$ is compatible to $y_{l_i}^{'}$ for any $i$ less than some $k$. Now we prove that $y_{k}^{*}$ is compatible to $y_{k}^{'}$. For contraction, suppose $y_{k}^{*}$ is not compatible to $y_{k}^{'}$. By the supposition, there is an event in $x$ that happens in $y_{k}^{'}$ while has not happened in $y_{l_k}^{*}$. However, $y_{l_{k-1}}^{*}$ is compatible to $y_{k-1}^{'}$, and, by pattern $A^{*}$, there is no node can be further activated in realization $\overline{x}$ by step $l_k^{*}$ under $S_{A^{*}}^G$. This implies that $S_{A^{'}}^G$ must wait for some null rounds between step $l_{k-1}^{'}$ and step $l_{k}^{'}$, which is a contradiction. 

By the construction of $S_{A^{*}}^G$, since $Node(S_{A^{*}}^{G^x})=Node(S_{A^{'}}^{G^x})$ in any f-realization $x$, we have $E[S_{A^{*}}^G]=E[S_{A^{'}}^G]$ according to Eq. (\ref{eq:e_def}).
\end{proof}

One can see that any strategy of a pattern other than $A^{*}$ cannot always simulate the one of pattern $A^{*}$ by the similar construction due to the feasibility issue as discussed above. Intuitively, pattern $A^{*}$ is the optimal because it maximizes the information obtained before making seeding decision and brings us more options in selecting seed nodes. The above result is summarized as follows.
\begin{theorem}
\label{theorem:result_1}
Pattern $A^{*}$ is the optimal pattern on any graph $G$, i.e., for any pattern $A$, $E[OPT_{A^*}^G] \geq E[OPT_{A}^G]$.
\end{theorem}
\begin{proof}
By Lemma \ref{lemma:result_1}, for any pattern $A$ and network $G$, we always have some strategy $S_{A^{*}}^G$ of $A^{*}$ such that $E[S_{A^{*}}^G] \geq E[OPT_{A}^G]$. Thus, 
\begin{equation}
E[OPT_{A^*}^G] \geq E[S_{A^{*}}^G] \geq E[OPT_{A}^G]. \nonumber
\end{equation}
\end{proof}

\subsection{Approximation Ratio}
In this section, we show that $\overline{S}_{A^{*}}^G$ has a approximation ratio of $(1-1/e)$.

The method to represent the random event space is critical to the analysis of a stochastic model. Essentially, the adaptive seeding strategy $\overline{S}_{A^{*}}^G$ forms a decision tree, where each node in the tree is a selected seed set and each out-edge of the tree-node represents a possible successive event. Let the root node be the first level. Then, each branch from level $i$ to level $i+1$ corresponds to a p-realization after round $i$ under $\overline{S}_{A^{*}}^{G}$. Each path from the root to a leave is formed by a sequence of p-realizations where each p-realization is compatible to its predecessor. For the decision tree of $\overline{S}_{A^{*}}^G$, let $Z_i=\{z_i^{1},...,z_i^{|Z_i|}\}$ be the set of the p-realizations (branches) from level $i$ to level $i+1$ where $|Z_i|$ is number of branches, and $Z_0=\{\epsilon\}$. Although the basic event space is unique, it can be represented via different decision trees under different strategies. For Example \ref{exp:intro} shown in Fig. \ref{fig:example}, the decision tree of a strategy of pattern $A^{*}$ on $G_1$ is shown in Fig. \ref{fig:tree_2} where the explanations are available in the caption. Note that for a DIC network G the decision tree of $\overline{S}_{A^{*}}^G$ is determinate. 


\begin{figure}[t]
\begin{center}
\includegraphics[width=3in]{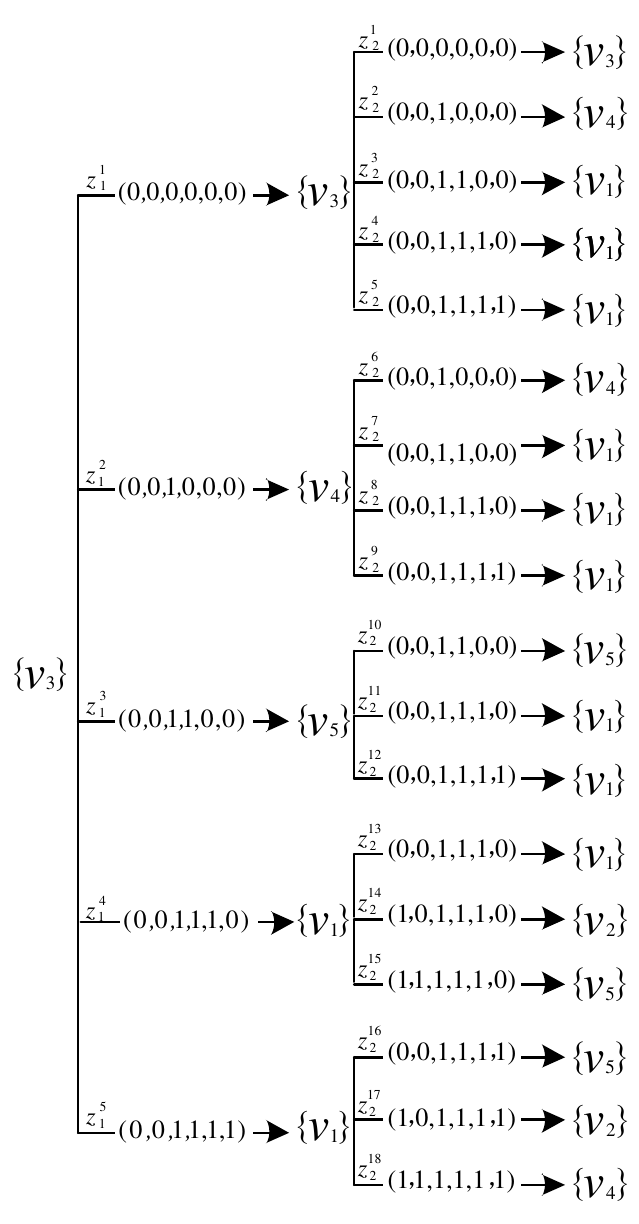} 
\end{center} 
\caption{\small{The decision tree of a strategy under pattern $A^*$ on the example DIC network $G_1$. For the vector $(x_1,x_2,x_3,x_4,x_5)$ on a branch $z_i^j$, that $x_i=0$ (resp. $x_i=1$) means node $v_i$ is active (resp. inactive) after round $i$ through branch $z_i^j$. In this example, branch $z_1^1$ implies $v_3$ is not successfully activated in step 1, and following pattern $A^{*}$ we have totally 5 and 18 branches from level 1 to level 2 and from level 2 to level 3, respectively.}}
\label{fig:tree_2}
\vspace{-4mm}
\end{figure}

Now we are ready to show the main result of this paper. Our goal is to prove that 
\begin{equation*}
E[OPT_{A^{*}}^G] \leq (1-1/e) \cdot E[\overline{S}_{A^{*}}^G]. 
\end{equation*}
For an arbitrary network $G$, let $t_i$ be the $i^{th}$ seed node  selected by $OPT_{A^{*}}^G$, and $T_i=\{t_1,...,t_i\}$. Similarly let $w_i$ be the $i^{th}$ seed node selected by $\overline{S}_{A^{*}}^G$ and $W_i=\{w_1,...,w_i\}$. Set $T_0=W_0=\emptyset$. We use the decision tree to analyze the seeding process. 

For a node set $V^{'}$ and a p-realization $z_i^j$, let
\begin{equation}
\label{eq:f_i^j}
F_i^j(V^{'}) = \sum_{x \in C_G(z_i^{j})} \text{Prob}[x|z_i^{j}] \cdot \textit{N}_{x}^G(V^{'})
\end{equation}
and
\begin{equation}
\label{eq:f_i}
F_i(V^{'})=\sum_{j=1}^{|Z_i|} \text{Prob}[z_i^j] \cdot F_i^j(V^{'}).
\end{equation}
One can see that $F_i^j(W_i)$ is the expected number of active nodes under seed set $W_i$ conditioned on p-realization $z_i^j$ and $F_B(W_B)=E[\overline{S}_{A^{*}}^G]$ . 

By Rule \ref{rule:greedy},
\begin{equation}
\label{eq:margin}
w_{i+1}= \argmax_v F_i^j(W_i \cup \{v\}). 
\end{equation}
Let
\begin{eqnarray}
\label{eq:delta_i^j}
\Delta_{i}^j= F_i^j(W_{i} \cup \{w_{i+1}\})-F_i^j(W_{i}),
\end{eqnarray}
for $0 \leq i \leq B-1$ .
\begin{lemma}
\label{lemma:sub}
\begin{eqnarray}
F_i^j(T_B) \leq F_i^j(W_i)+B \cdot \Delta_{i}^j \nonumber
\end{eqnarray}

\end{lemma} 
\begin{proof}
For $1 \leq h \leq B$, by Property \ref{property:submodular}, 
\begin{eqnarray*}
\textit{N}_{x}^{G}(T_h \cup W_i)&-&\textit{N}_{x}^{G}(T_{h-1} \cup W_i) \\
&\leq& \textit{N}_{x}^{G}(\{t_h\} \cup W_i)-\textit{N}_{x}^{G}(W_i) 
\end{eqnarray*}
Thus, 
\begin{eqnarray}
& &\sum_{x \in C_G(z_i^{j})} \text{Prob}[x|z_i^{j}] \big( \textit{N}_{x}^{G}(T_h \cup W_i)-\textit{N}_{x}^{G}(T_{h-1} \cup W_i) \big) \nonumber \\ 
&\leq& \sum_{x \in C_G(z_i^{j})} \text{Prob}[x|z_i^{j}] \big(\textit{N}_{x}^{G}(\{t_h\} \cup W_i)-\textit{N}_{x}^{G}(W_i) \big) \nonumber \\
&&\{~\text{\small by Eq. (\ref{eq:f_i^j})}~\} \nonumber \\
&=& F_i^j(\{t_h\} \cup W_i)) -F_i^j(W_i)\nonumber \\
&&\{~\text{\small by Eq. (\ref{eq:margin})}~\} \nonumber \\
&\leq& F_i^j(W_{i+1})) -F_i^j(W_i)\nonumber \\
&&\{~\text{\small by Eq. (\ref{eq:delta_i^j})}~\} \nonumber \\
&=& \Delta_{i}^j.\nonumber
\end{eqnarray}
Adding the above inequalities for all $1 \leq h \leq B$, we have  
\begin{eqnarray*}
& &\sum_{1 \leq h \leq B} \sum_{x \in C_G(z_i^{j})} \text{Prob}[x|z_i^{j}] \big( \textit{N}_{x}^{G}(T_h \cup W_i)-\textit{N}_{x}^{G}(T_{h-1} \cup W_i) \big) \nonumber \\
&=&   \sum_{x \in C_G(z_i^{j})} \text{Prob}[x|z_i^{j}] \big( \textit{N}_{x}^{G}(T_B \cup W_i)-\textit{N}_{x}^{G}(T_{0} \cup W_i) \big)\nonumber \\
&=& F_i^j(T_B \cup W_i)-F_i^j(W_i) \\ \nonumber
&\leq& B \cdot \Delta_{i}^j.\nonumber
\end{eqnarray*}
Thus,
\begin{eqnarray}
&&F_i^j(T_B) 
\leq F_i^j(T_B \cup W_i) \leq F_i^j(W_i) + B \cdot \Delta_{i}.
\end{eqnarray}
Note that $T_B$ depends on $x$ and $W_i$ depends on $z_i^j$.

\end{proof}

Let 
\begin{equation}
\label{eq:delta_i}
\Delta_i= \sum_{j=1}^{|Z_i|} \text{Prob}[z_i^j] \cdot \Delta_i^j.
\end{equation} 
\begin{lemma}
\label{lemma:expand}
$F_{i}(W_i)=\Delta_{0}+...+\Delta_{i-1}$.
\end{lemma} 
\begin{proof}
Note that, for any $0 \leq h<B$
\begin{eqnarray*}
F_{h}(W_{h})&=&\sum_{j=1}^{|z_{{h}}|} \text{Prob} \cdot [z_{{h}}^j]F_i^j(W_{{h}}) \\
&=& \sum_{j=1}^{|z_{{h}-1}|} \text{Prob}[z_{{h}-1}^j] \cdot F_{{h}-1}^j(W_{{h}-1} \cup \{w_{{h}}\}).\nonumber
\end{eqnarray*}
Thus, we have
\begin{eqnarray}
& &\Delta_{0}+...+\Delta_{i-1}\nonumber \\
&&\{~\text{\small by Eq. (\ref{eq:delta_i})}~\} \nonumber \\
&=& \sum_{h<i} \sum_{j=1}^{|Z_h|} \text{Prob}[z_h^j] \cdot \Delta_h^j \nonumber \\
&&\{~\text{\small by Eq. (\ref{eq:delta_i^j})}~\} \nonumber \\
&=& \sum_{h<i} \sum_{j=1}^{|Z_h|} \text{Prob}[z_h^j] \cdot (F_h^j(W_{h} \cup \{w_{h+1}\})-F_h^j(W_{h}))\nonumber \\
&=&\sum_{h<i}(F_{h+1}(W_{h+1})-F_{h}(W_{h})) \nonumber \\
&=& F_i(W_i)-\sum_{j=1}^{|z_{0}|} \text{Prob}[z_{0}^j] \cdot F_0^j(W_{0}) \nonumber \\  
&=&  F_i(W_i) \nonumber
\end{eqnarray}
\end{proof}

Finally, we have the following lemma.
\begin{lemma}
\label{lemma:result_2}
$E[OPT_{A^{*}}^G] \leq (1-1/e) E[S_{A^{*}}^G]$ 
\end{lemma}

\begin{proof}
For $0\leq i \leq B-1$, we have
\begin{equation}
E[OPT_{A^{*}}^G]=F_0^1(T_B)=F_i(T_B). \nonumber
\end{equation}
By Lemma \ref{lemma:sub},
\begin{equation}
F_i^j(T_B) \leq F_i^j(W_i)+B \cdot \Delta_{i}^j, \nonumber
\end{equation}
i.e.,
\begin{equation}
E[OPT_{A^{*}}^G] \leq F_i(W_i)+B \cdot \Delta_{i}, \nonumber
\end{equation}
Thus, combining Lemma \ref{lemma:expand}, 
\begin{equation}
\label{eq:1}
E[OPT_{A^{*}}^G] \leq \Delta_0+...+\Delta_{i-1}+B \cdot \Delta_{i}, 
\end{equation}
By multiplying the both sides of Eq. (\ref{eq:1}) by $(1-1/B)^{B-1-i}$ we have
\begin{eqnarray}
\label{eq:2}
&&E[OPT_{A^{*}}^G] \cdot (1-1/B)^{B-1-i}\\ \nonumber
&\leq& (\Delta_0+...+\Delta_{i-1}+B \cdot \Delta_{i})  \cdot (1-1/B)^{B-1-i} 
\end{eqnarray}
Now we add up Eq. (\ref{eq:2}) for $0 \leq i \leq B-1$. The left side of the summation is 
\begin{eqnarray}
\label{eq:3}
\sum_{i=0}^{B-1}E[OPT_{A^{*}}^G] \cdot (1-1/B)^{B-1-i} \nonumber \\ 
= B(1-(1-\frac{1}{B})^{B}) \cdot E[OPT_{A^{*}}^G]
\end{eqnarray}
On the right side, the coefficient of $\Delta_{i}$ is
\begin{eqnarray}
\label{eq:4}
&& B \cdot (1-\dfrac{1}{B})^{B-i}+\sum_{j=i}^{B-1} (1-1/B)^{B-1-j}=B 
\end{eqnarray}
Thus, by Eqs. (\ref{eq:3}) and (\ref{eq:4}),
\begin{eqnarray}
\label{eq:5}
&& E[OPT_{A^{*}}^G]\cdot B \cdot (1-(1-1/B)^B) \\
&\leq& B \cdot (\Delta_0+...+\Delta_{B-1}) \nonumber \\
&&\{~\text{\small by Lemma \ref{lemma:expand}}~\} \nonumber \\
&=&B \cdot F_B(W_B) \nonumber \\
&=&B\cdot E[S_{A^{*}}^G]. \nonumber 
\end{eqnarray}
Therefore, the approximation ratio of $S_{A^{*}}^G$ is at least $(1-1/e)$.
\end{proof}
The above result is summarized as follows.
\begin{theorem}
\label{theorem:result_2}
$\overline{S}_{A^{*}}^{G}$ is a strategy within a factor $1-1/e$ from the optimal strategy of pattern $A^{*}$.
\end{theorem}

Since $A^*$ is the optimal pattern as discussed in Sec. \ref{subsec:opt_pattern}, $\overline{S}_{A^{*}}^{G}$ is an $(1-1/e)$-approximation of AIM problem.
\begin{corollary}
\label{coro:result}
$\overline{S}_{A^{*}}^{G}$ is an $(1-1/e)$-approximation of AIM problem.
\end{corollary}

Golovin \textit{et al.} \cite{golovin2010adaptive} apply the stochastic submodular maximization technique to several applications including the influence diffusion in social networks. They conjecture that applying Rule. \ref{rule:greedy} to pattern $A_0$ in the classic IC model yields an $(1-1/e)$-approximation to the optimal seeding strategy under pattern $A_0$. Actually the derivation of Theorem \ref{theorem:result_2} can be applied to any pattern where we seed at most one node in each step in the DIC model. Therefore, since the classic IC model is a special case of the DIC model, the truth of their conjecture in \cite{golovin2010adaptive} can be verified. In fact, under any pattern, Rule \ref{rule:greedy} is able to provide an approximation with the same ratio. As this paper focuses on designing practical seeding strategies, we will not show the technical proof of that result. 
\begin{algorithm}[b]
\caption{ \textbf{A-Greedy}}\label{alg:a^*}
\begin{algorithmic}[1]
\State \textbf{Input}: \small{$G=(V,E,F_V,F_E)$ and budget B.}
\State CurrentBudget $\leftarrow 0$; $A \leftarrow \emptyset$; 
\State $y_0=\epsilon$;~~// $y_i$ is the p-realization after round $i$ .
\For {each $v$ in $V$}  $S_v \leftarrow +\infty$;
\EndFor
\For {$i=1:N$}	
	\If {(CurrentBudget$<$B and no nodes can be further activated) }
		\For {each $v$ in $V\setminus A$}  $s_v \leftarrow false$;\EndFor
		\While {true}
		\State $v^{*}=\argmax_{v \in V\setminus A}S_v $
			\If {($s_{v^{*}}=true$)} $A \leftarrow A \cup v^{*}$; break;
			\Else~~{$s_{v^{*}}=\sum_{x \in C_G(y_{i-1})} \text{Prob}[x|y_{i-1}] \cdot \textit{N}_{x}^G(A \cup {v^{*}})$}
			\EndIf
		\EndWhile
		\State CurrentBudget+CurrentBudget+1;
	\EndIf
	\State Get $y_i$;~~// wait for a round of spread
\EndFor
\State $y^{*}\leftarrow y_{N}$
\State Return $\textit{N}_{y^*}^G(A)$
\end{algorithmic}
\end{algorithm}
\subsection{Implementation Issues}
To implement the proposed greedy algorithm, the only problem left is to calculate Eq. (\ref{eq:margin}). Unfortunately, as discussed in \cite{chen2009efficient}, it is \#P-hard to calculate the real value of  $ \sum_{x \in C_G(z_i^{j})} \text{Prob}[x|y] \cdot \textit{N}_{x}^G(V^{'})$ in Eq. (\ref{eq:f_i^j}). However, we can employ the Monte Carlo simulation to obtain an accurate estimation. By the Hoeffding's Inequality, the error of the estimation can be infinitely small when a sufficient number of simulations are performed. Another issue one may concern is the efficiency of the greedy algorithm because a large number of simulation may required for an accurate estimation. As shown in \cite{leskovec2007cost}, the Lazy-Forward technique could be implemented in a hill-climbing strategy and leads to far fewer evaluations. The pseudo-code of $\overline{S}_{A^{*}}^{G}$  with Lazy-Forward method is shown in Algorithm \ref{alg:a^*}. We denote this adaptive seeding strategy by \textbf{A-Greedy}.

\section{Heuristic Seeding Strategy}
\label{sec:heuristic}
In this section, we present a heuristic adaptive seeding strategy based on the greedy algorithm in Sec. \ref{sec:greedy}. To reduce the time consumed in the seeding process, a simple idea is to reduce the number of nodes that could be considered as seed nodes. Obviously, the performance of the seeding strategy cannot be guaranteed if we inappropriately exclude some nodes before the seeding process. Thus, we aim to study that what kinds of nodes can be ignored in the seeding process. An important observation as shown later in Sec. \ref{sec:exp} it that there could be a significant gap of the strength between the influential nodes and other nodes. This fact is coincident to the power-law nature of the real-world social networks where degree of the nodes follows the exponential distribution. Motivated by this observation, we design a heuristic seeding strategy, termed as \textbf{H-Greedy}, that narrows the candidate seed set before the seeding process.

\textbf{H-Greedy.} Let $H(v)$ be the number of the nodes can be activated by a single seed node $v$. Let E[.] and Std[.] denote the mean and the standard deviation of a random variable. H-Greedy consists of two steps. First, before we start the seeding process, by Monti Carlo simulation, we first obtain the estimates of  $E[H(v)]$, $E[\sum_{v \in V}H(v)/N]$, and $Std[\sum_{v \in V}H(v)/N]$. We denote those three estimates by $\text{\^{E}}[H(v)]$, $\text{\^{E}}[\sum_{v \in V}H(v)/N]$, and $\text{\^{Std}}[\sum_{v \in V}H(v)/N]$, respectively. Then, when determining a seed node in the seeding process, we omit a node $v$ if $E[H(v)]$ is less than the lower 1-sigma control\footnote{Mean minus standard deviation} of $\sum_{v \in V}H(v)/N$. 

As discussed in the prior works, we used to execute Monte Carlo simulation for 10000 to 20000 times for an accurate estimation. However, in the first step of H-Greedy, 1000 to 2000 simulations are sufficient. This is because the estimates are not necessary to be very accurate as they are merely used to narrow the candidate set of seed nodes. With a smaller set of candidate seed nodes the time consumed in the seeding process can be significantly reduced as about a half of the nodes will not be considered to be seed nodes. A shown later, the performance of H-Greedy is closed to Greedy which has a provable performance guarantee. We will further discuss the feasibility of H-Greedy in the next section.

\section{Experiment}
\label{sec:exp}
In this section, we show the results of the conducted experiments. In order to evaluate the proposed adaptive seeding strategies, we examine the performance of our strategies from the following aspects: (a) the influence spread comparing to non-adaptive seeding strategies; (b) the effectiveness and efficiency of the heuristic strategy.

\begin{figure*}[t]
\captionsetup{justification=centering}
\subfloat[$\mathscr{F}^1$ with Prob$\text{[}X_u=1\text{]}=1$ \newline on Hep ]{\label{fig:hep_1_1}\includegraphics[width=0.24\textwidth]{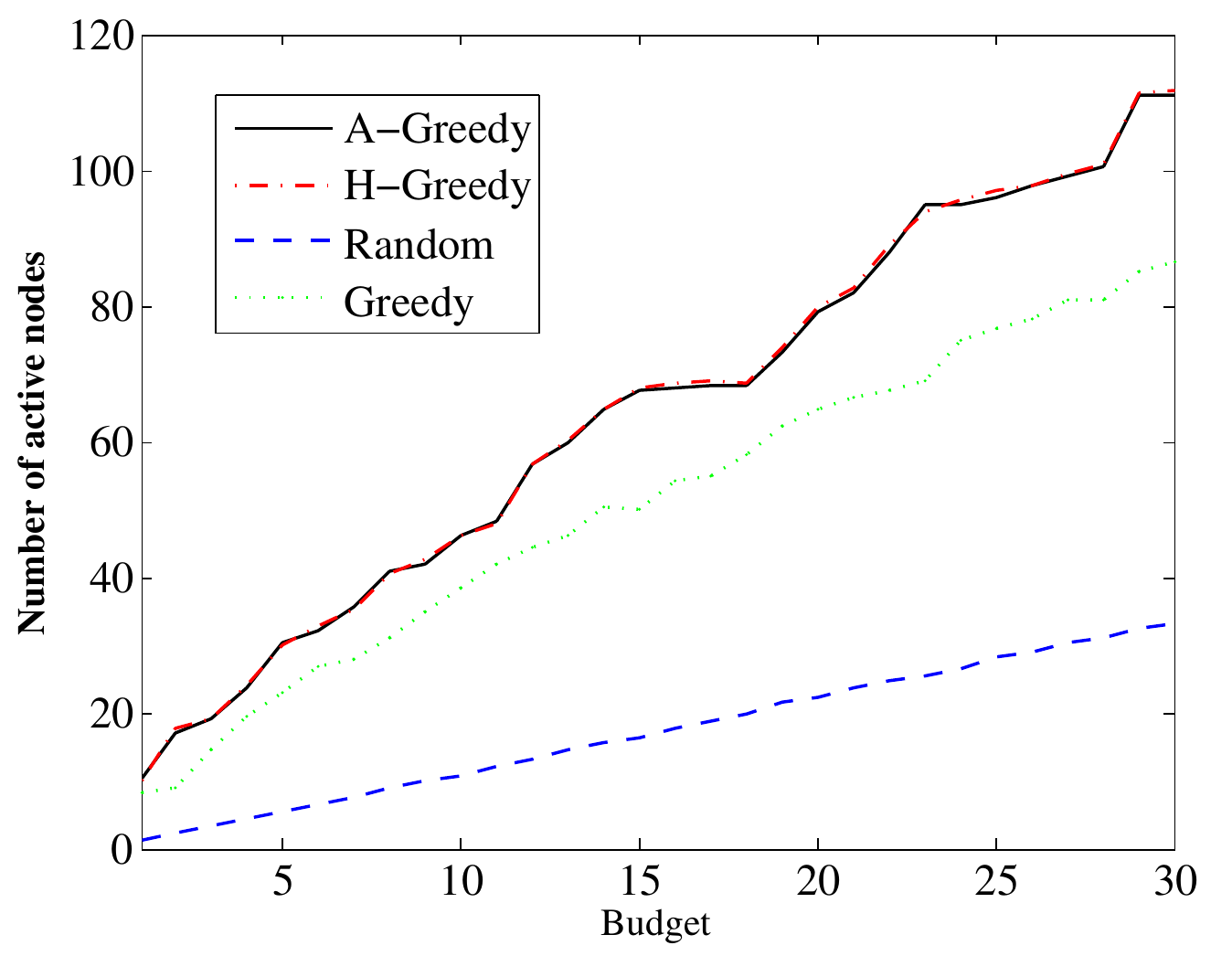}} \hspace{0mm}
\subfloat[$\mathscr{F}^3$ with Prob$\text{[}X_u=1\text{]}=1$ \newline on Hep ]{\label{fig:hep_1_5}\includegraphics[width=0.24\textwidth]{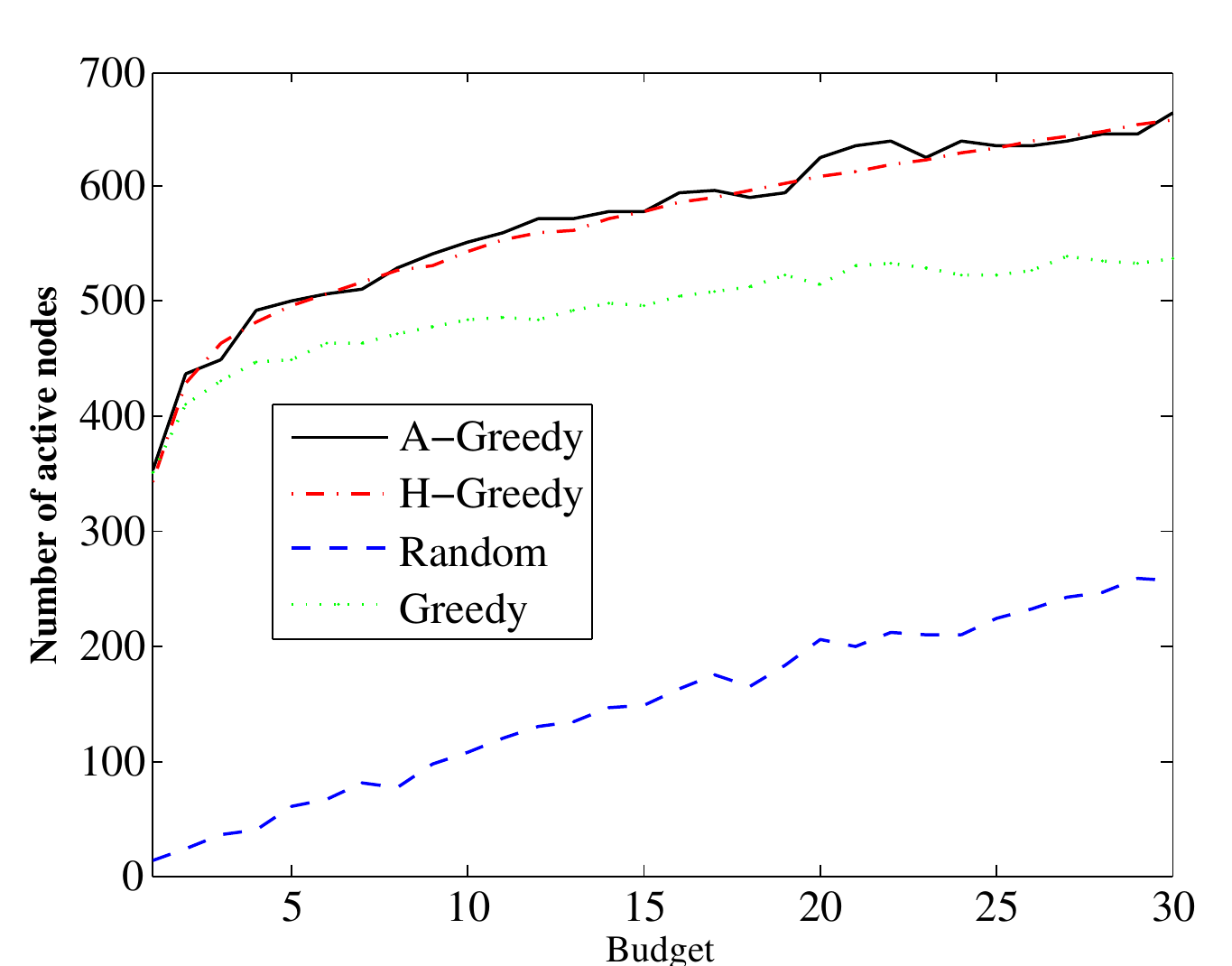}} \hspace{0mm}
\subfloat[$\mathscr{F}^1$ with Prob$\text{[}X_u=1\text{]}=0.5$ \newline on Hep]{\label{fig:hep_5_1}\includegraphics[width=0.24\textwidth]{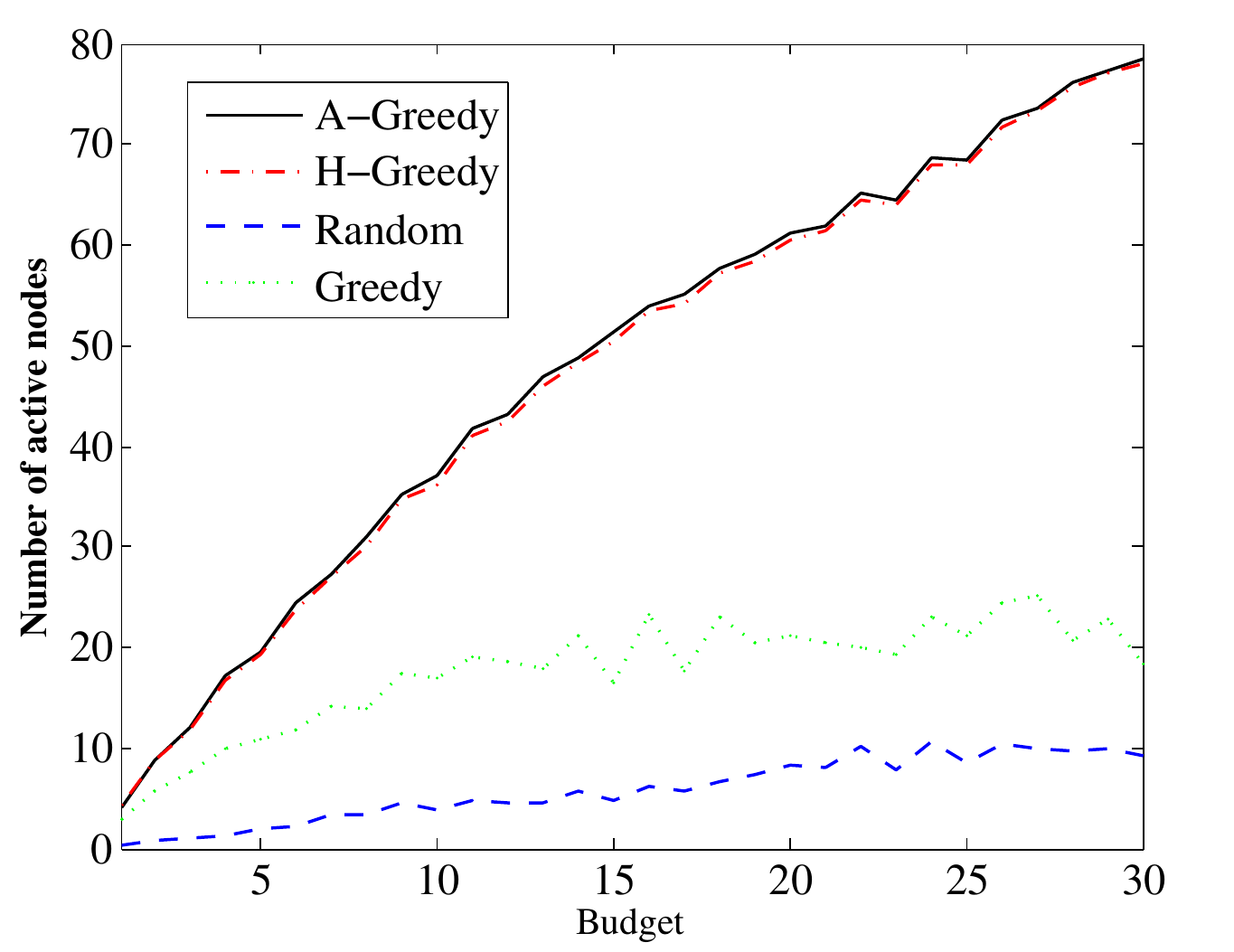}} \hspace{0mm}
\subfloat[$\mathscr{F}^2$ with Prob$\text{[}X_u=1\text{]}=0.5$ on Hep]{\label{fig:hep_5_3}\includegraphics[width=0.24\textwidth]{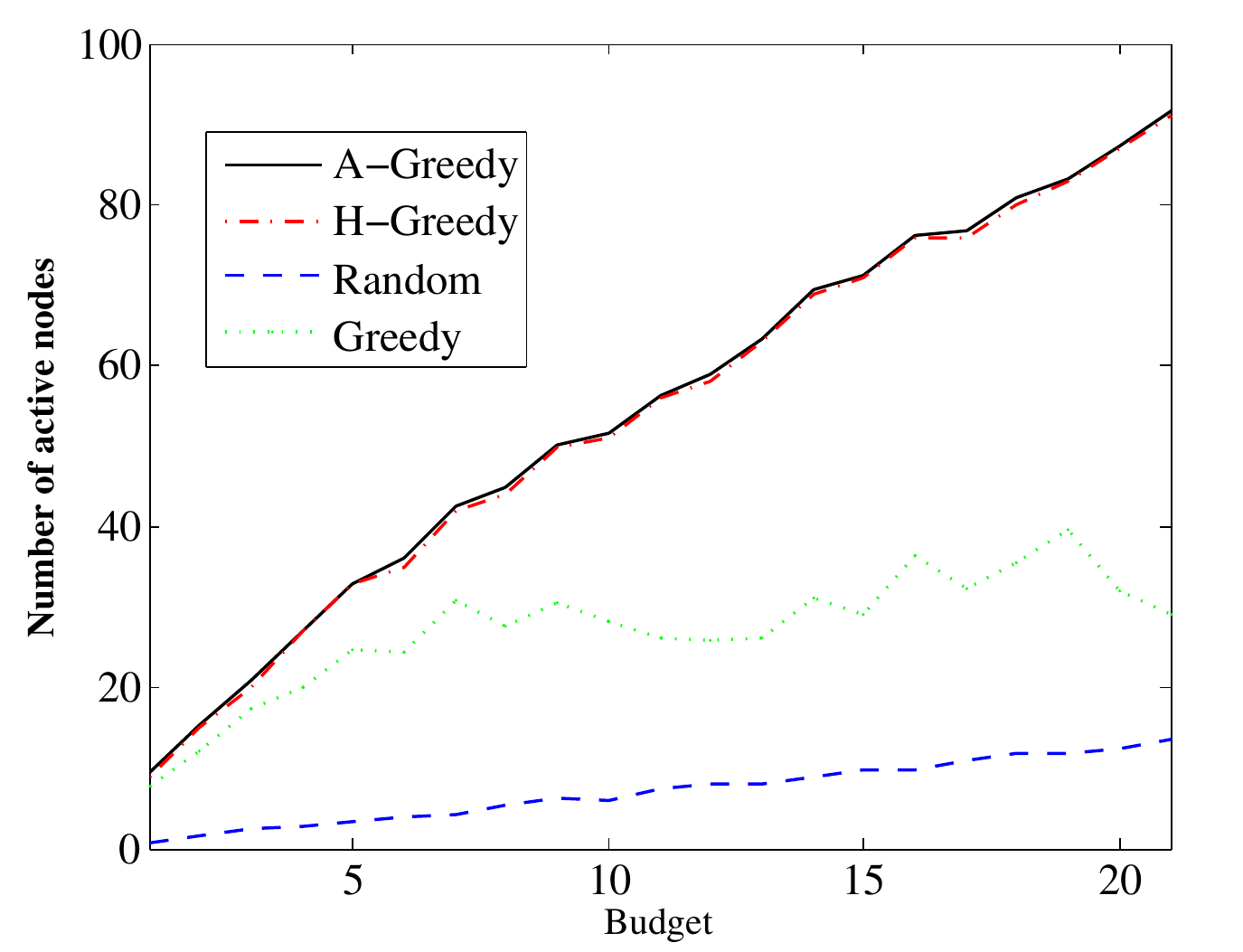}} \hspace{0mm}

\subfloat[$\mathscr{F}^1$ with Prob$\text{[}X_u=1\text{]}=1$  on PL]{\label{fig:pl_1}\includegraphics[width=0.32\textwidth]{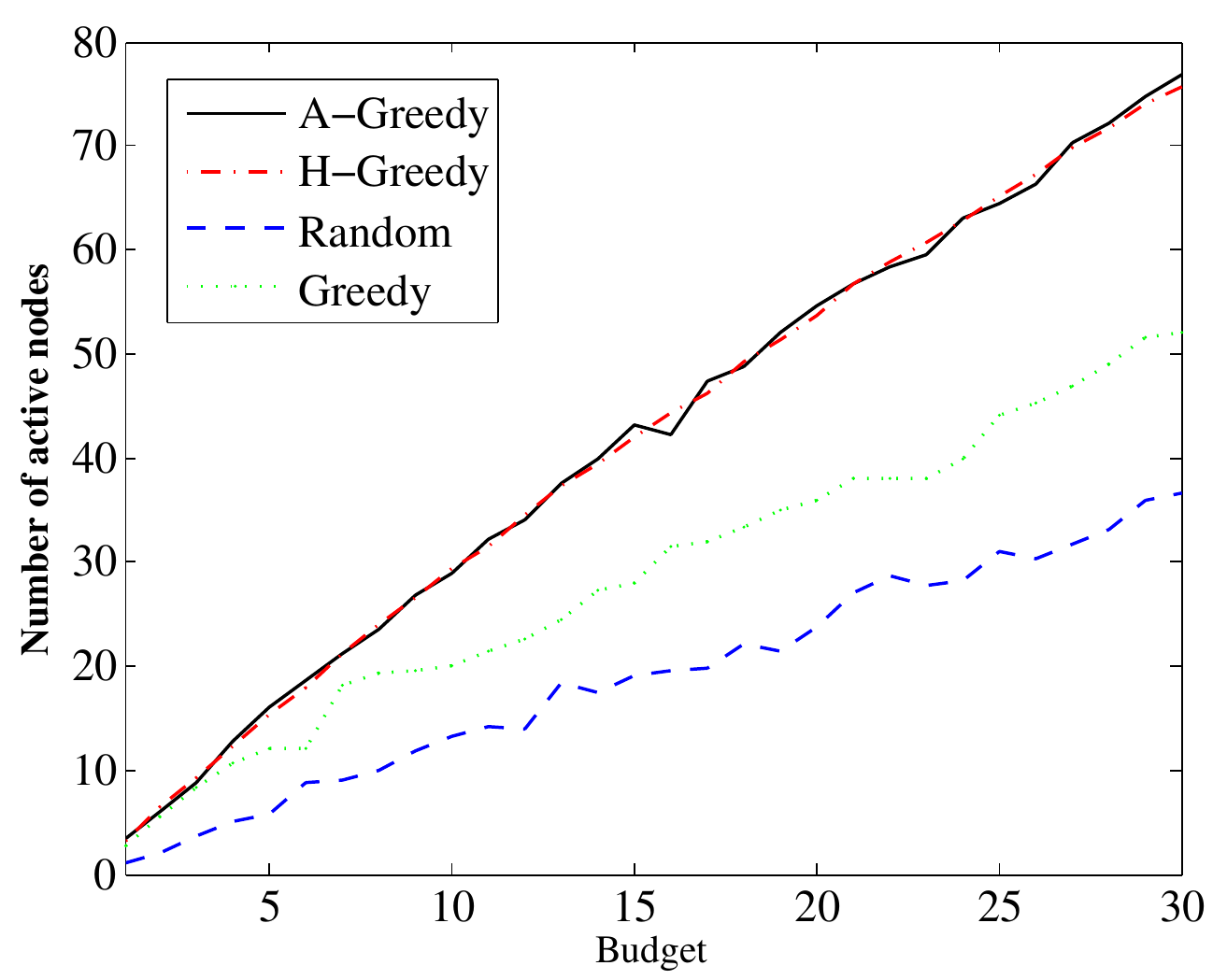}} \hspace{0mm}
\subfloat[$\mathscr{F}^3$ with Prob$\text{[}X_u=1\text{]}=0.5$ on PL]{\label{fig:pl_2}\includegraphics[width=0.32\textwidth]{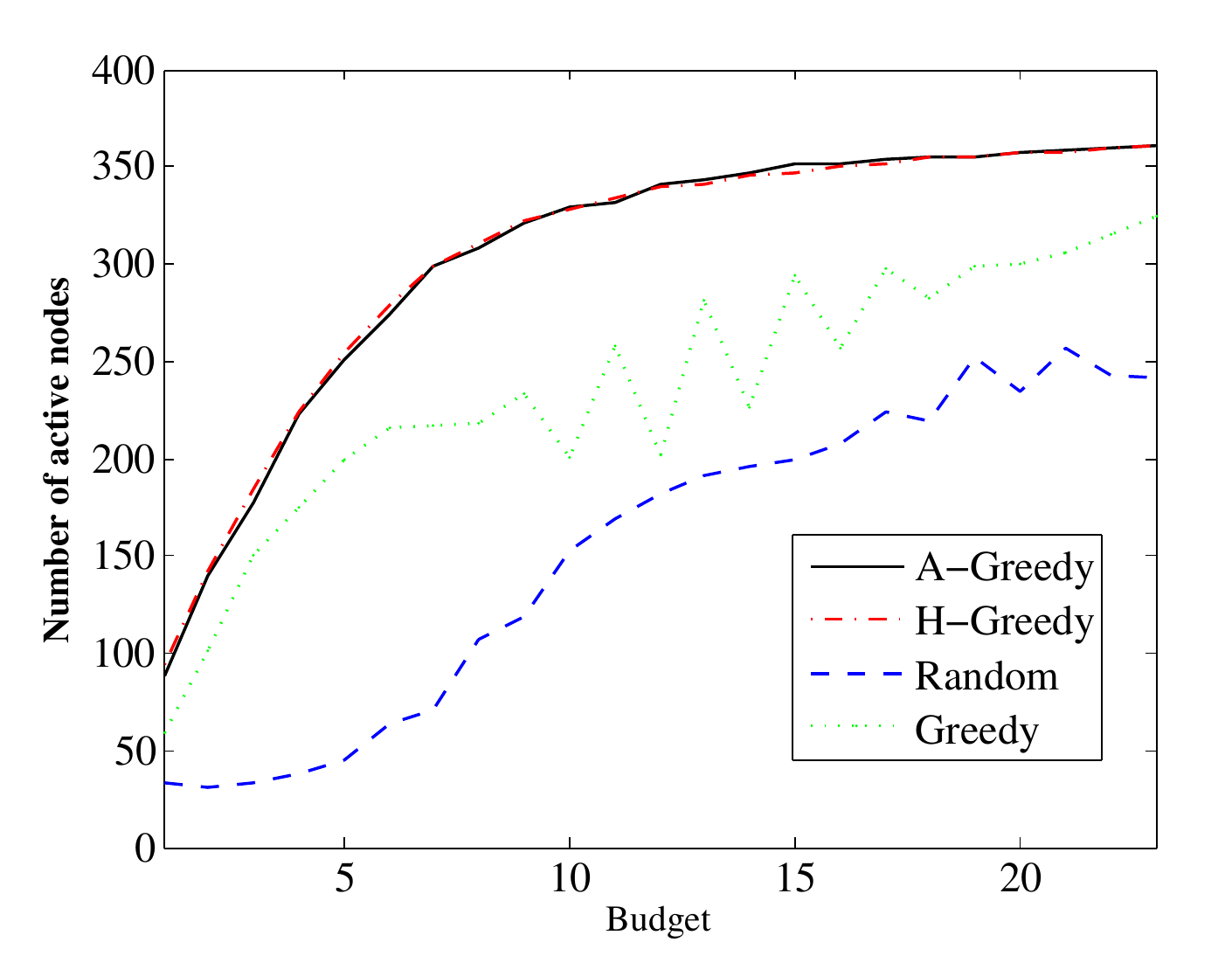}} \hspace{0mm}
\subfloat[$\mathscr{F}^1$ with Prob$\text{[}X_u=1\text{]}=1$ on Wiki]{\label{fig:wiki_1}\includegraphics[width=0.32\textwidth]{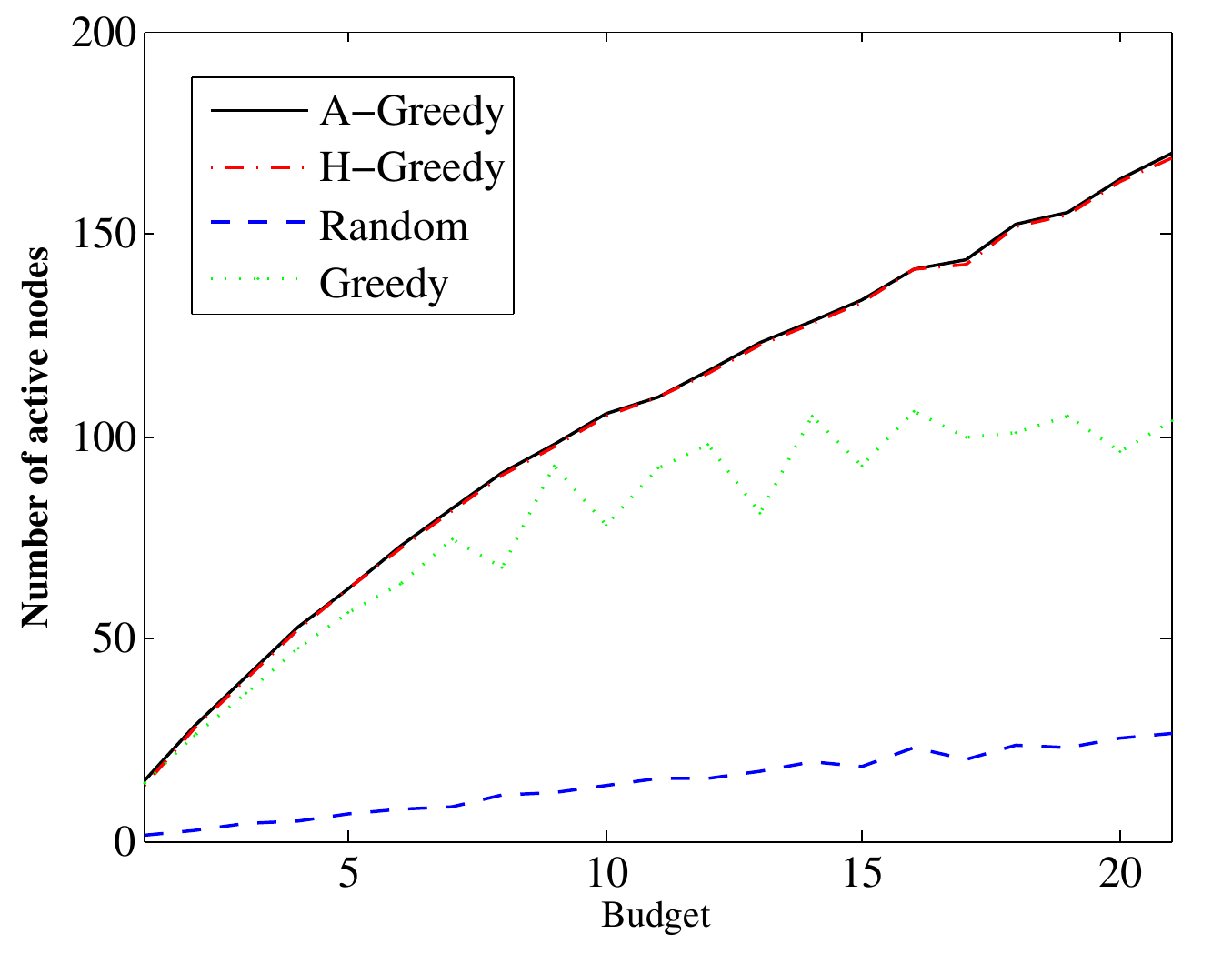}} \hspace{0mm}

\caption{\small Comparing A-Greedy with Greedy. In all seven graphs, the y-axis and x-axis denote the number of active nodes and the budget, respectively. Each graph gives
four curves plotting the influence spread under four seeding strategies, respectively. } 
\vspace{-2mm} \normalsize
\label{fig:a-greedy}
\end{figure*}

\subsection{Experiment Setup}
\label{subsec:exp_set}
In order to fairly compare the performance of our seeding strategies to that of the existing approaches, we employ two real-world social networks, which have been widely used in the prior works, and a synthetic power-law network which is able to capture the key features of real social networks. The propagation probabilities are generated from three distributions, as shown later.

\textbf{Network structure.} The first real-world social network, denoted by \textbf{Hep}, is an academic collaboration from co-authorships in physics. Hep is compiled from the "High Energy Physics - Theory" section of the e-print arXiv\footnote{http://www.arXiv.org} and has been widely used in the prior works (e.g. \cite{kempe2003maximizing,chen2009efficient,long2011minimizing} and \cite{zhang2014minimizing}). For each pair of authors who has a co-authorship, we have two directed edges from each one to the other. The resulting network has about 15,000 nodes and 58,000 directed edges. The second dataset, denoted by \textbf{Wiki}, contains the Wikipedia voting data \cite{leskovec2009wikipedia} from the inception of Wikipedia. Nodes in this network represent wikipedia users and a directed edge from node $u$ to node $v$ represents that user $u$ votes on user $v$, which mean $v$ has influence over $u$. Thus, if there is an edge from $u$ to $v$ in the original data, we add an edge from $v$ to $u$ in Wiki. Wiki has about 8,600 nodes and 103,000 directed edges and has been studied in \cite{chen2011influence}, \cite{li2014influence} and \cite{li2013rumor}. The last dataset is a synthetic power-law network generated by \cite{cowendigg}. The synthetic power-law network selected in this paper, denoted by \textbf{PL}, includes 2500 nodes and 26,000 directed edges. Power-law degree distribution has been shown to be one of the most important characteristics of social networks \cite{clauset2009power}. We use PL dataset to evaluate the performance of the proposed seeding strategies in general social networks.

\textbf{Propagation probability.} The three distributions $\mathscr{F}^i (i=1,2,3)$ of the propagation probability $X_e$ of an edge $e$ are shown as follows. In $\mathscr{F}^1$, the propagation probability are fixed as 0.01, which is the same as that in \cite{kempe2003maximizing}. $\mathscr{F}^2$ is an exponential distributions with a mean of $0.01$. $\mathscr{F}^3$ is a uniform discrete distribution over $\{0.1,0,01,0,001\}$. 

\textbf{Activation probability.} We assign a uniform activation probability on each node $u$, choosing $\text{Prob}[X_u=1]$ to be 1 and 0.5.

Note that it reduces to the classic IC model if $\mathscr{F}^1$ and $X_u=1$. 

\textbf{Seeding strategies.}  The tested seeding strategies are shown as follows.
\begin{enumerate}
\item {\textbf{Greedy}.} This is the state-of-art non-adaptive seeding approach proposed in \cite{kempe2003maximizing}. In Greedy, the nodes are selected by a hill-climbing algorithm before the diffusion process. When implementing Greedy in the DIC model, we fixed the propagation probability by its mean as the real propagation probabilities are unavailable in the DIC model before the start of diffusion process. For each estimation, 10000 simulations are run to obtain an accurate estimate.
\item{\textbf{A-Greedy}.} This is the greedy adaptive seeding strategy proposed in Sec. \ref{sec:greedy}. Similarly, 10000 simulations are run to obtain an accurate estimate of $\sum_{x \in C_G(y_{i-1})} \text{Prob}[x|y_{i-1}] \cdot \textit{N}_{x}^G(A \cup {v^{*}})$ in line 11 of Algorithm  \ref{alg:a^*}.
\item{\textbf{H-Greedy}.} This is the heuristic adaptive seeding strategy proposed in Sec. \ref{sec:heuristic}. In the first step of H-Greedy, 2000 simulations are run to obtain the estimates mentioned in Sec. \ref{sec:heuristic}. 
\item{\textbf{Random}.} This is a baseline seeding strategy where the seed nodes are selected randomly. 
\end{enumerate}
As discussed in the prior works, the seeding strategies based on the shortest-path and high-degree perform worst than Greedy. Thus we ignore other measures. In our experiment, the budget is chosen from 10 to 30.

\begin{figure*}[t]
\subfloat[$\mathscr{F}^1$ with Prob$\text{[}X_u=1\text{]}=1$]{\label{fig:path_1}\includegraphics[width=0.32\textwidth]{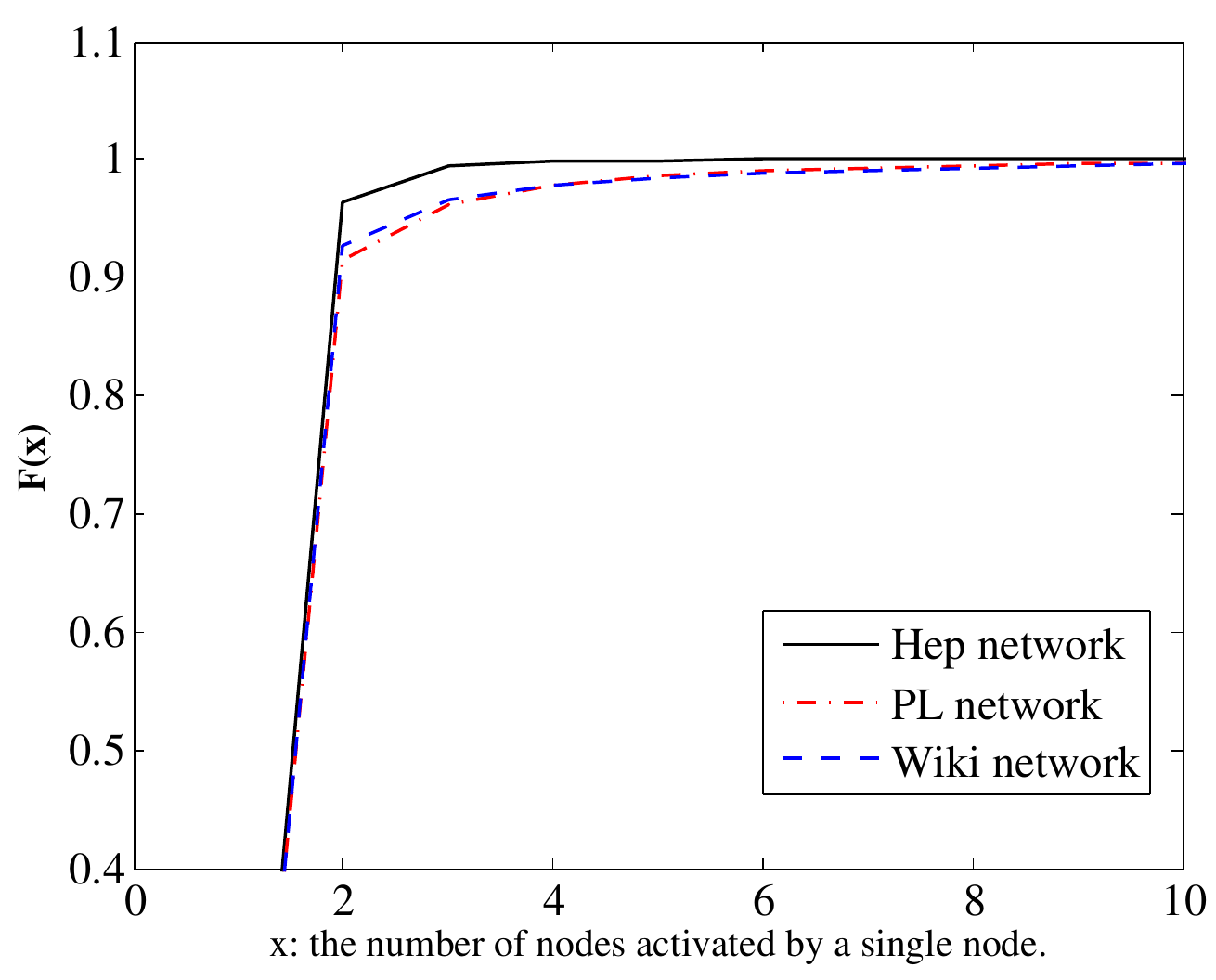}} \hspace{0mm}
\subfloat[$\mathscr{F}^2$ with Prob$\text{[}X_u=1\text{]}=1$]{\label{fig:path_2}\includegraphics[width=0.32\textwidth]{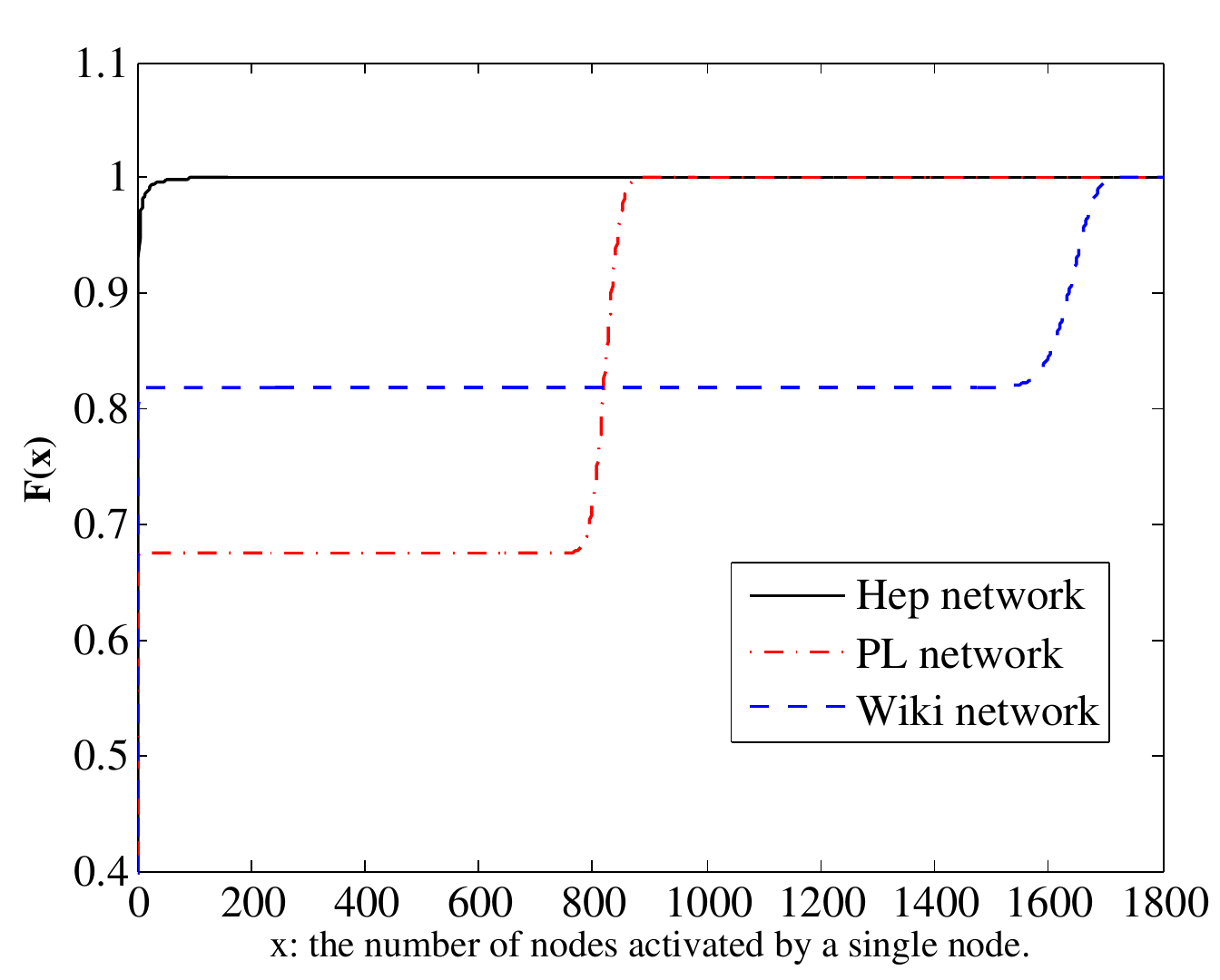}} \hspace{0mm}
\subfloat[$\mathscr{F}^3$ with Prob$\text{[}X_u=1\text{]}=1$]{\label{fig:path_3}\includegraphics[width=0.32\textwidth]{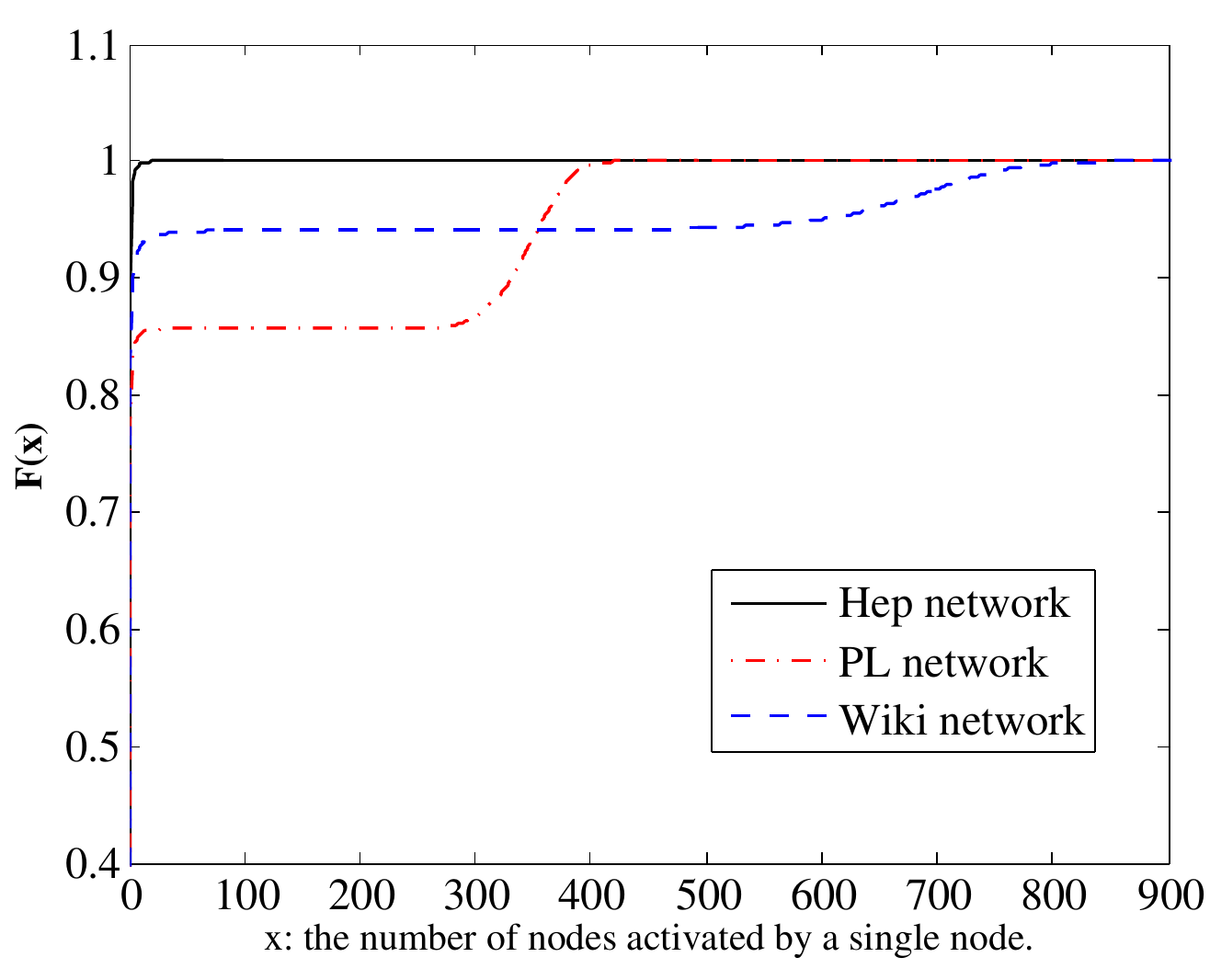}} \hspace{0mm}

\caption{Distributions of $E(H(v))$ of the three datasets under different propagation probability.} 
\vspace{-2mm} \normalsize
\label{fig:pathtest}
\end{figure*}

\begin{figure*}[t]
\subfloat[$\mathscr{F}^2$ with Prob$\text{[}X_u=1\text{]}=1$ on PL]{\label{fig:pl_1_4_h}\includegraphics[width=0.45\textwidth]{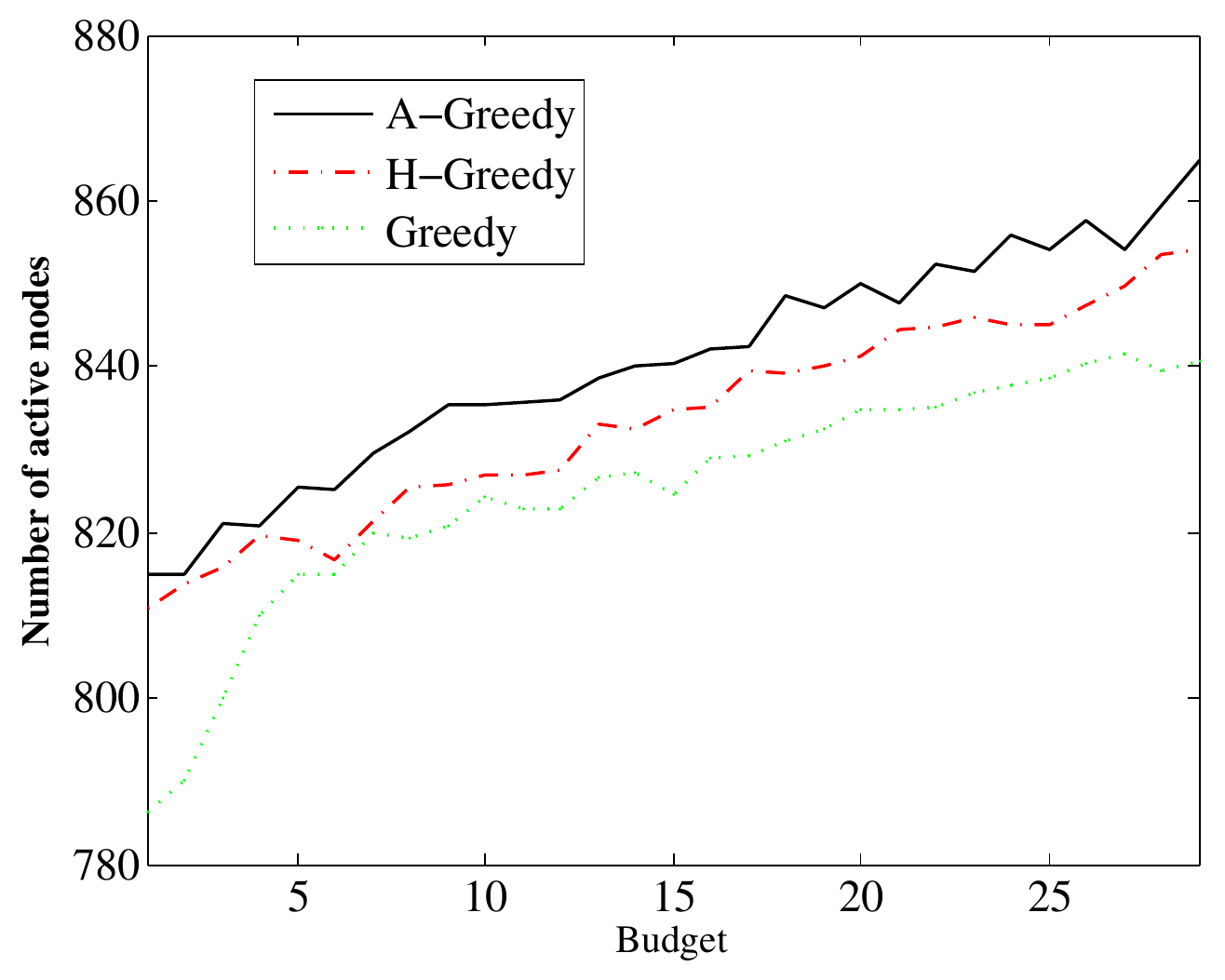}} \hspace{10mm}
\subfloat[$\mathscr{F}^3$ with Prob$\text{[}X_u=1\text{]}=1$ on PL]{\label{fig:pl_5_5_h}\includegraphics[width=0.45\textwidth]{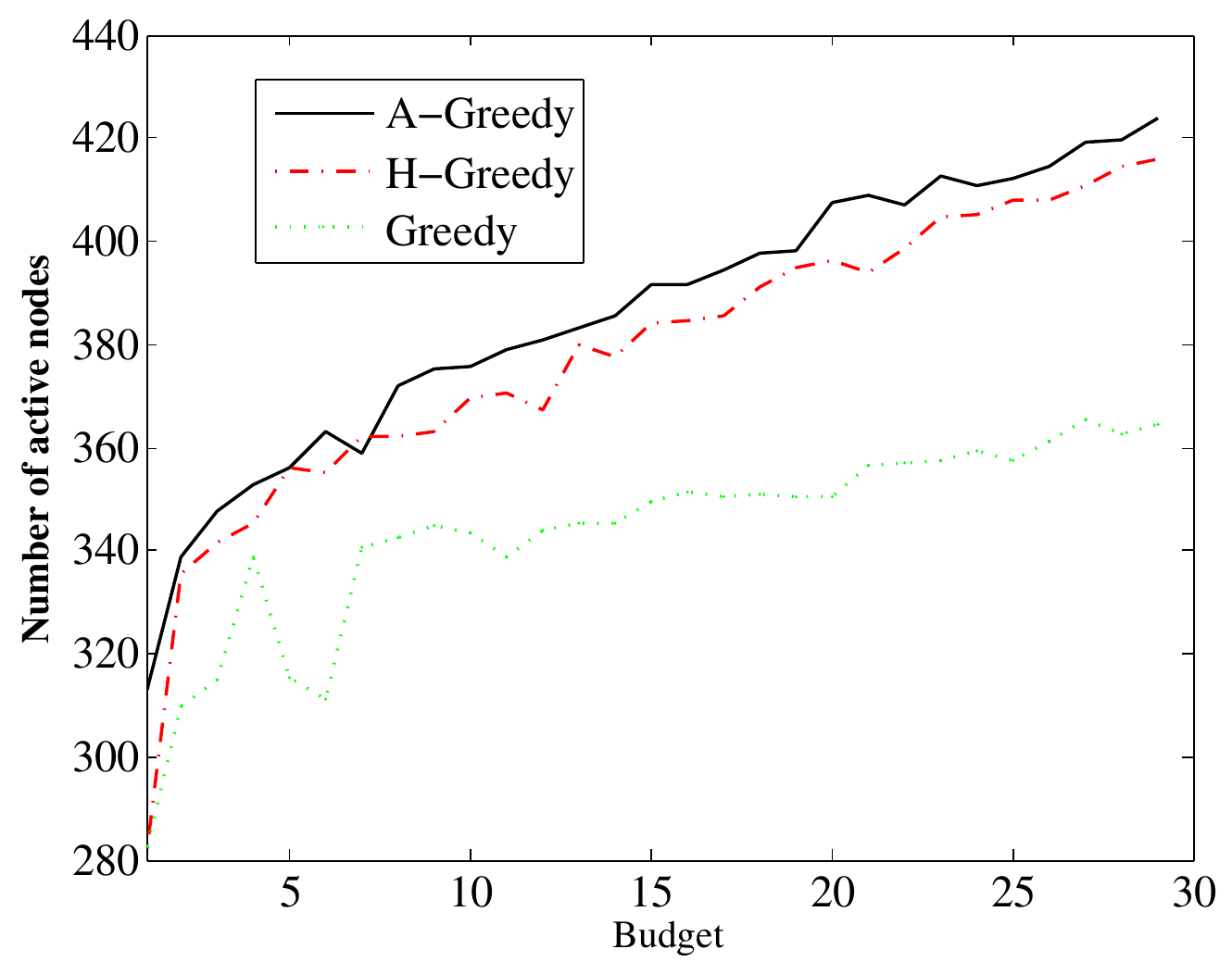}} \hspace{0mm}

\caption{\small Comparing H-Greedy with A-Greedy. The y-axis and x-axis denote the number of active nodes and the budget, respectively. Each graph gives
three curves plotting the influence spread under A-Greedy, H-Greedy and Greedy, respectively. We ignore Random here as it performs poorly.} 
\vspace{-2mm} \normalsize
\label{fig:h-greedy}
\end{figure*}

\subsection{Results}
\label{subsec:results}
First, we discuss the performance of A-Greedy. As shown in Fig. \ref{fig:a-greedy}, A-Greedy outperforms Greedy under all circumstances. This is intuitive as the adaptive seeding strategies are able to utilize the outcomes of the past rounds. As shown in Fig. \ref{fig:hep_1_1},  A-Greedy is superior to Greedy by a notable margin even in the classic IC model. For the DIC model where the diffusion process is of more uncertainness, the results herein verify the significant advantages of the adaptive seeding strategy over the non-adaptive seeding strategy. We discuss the results in detail in the following. 

For the Hep network, as shown in Fig. \ref{fig:hep_1_1}, A-Greedy is 125\% better than Greedy in the classic IC model under $\mathscr{F}^1$ with $\text{Prob}[X_u=1]=1$. While the uncertainness of the diffusion process getting increased, namely by changing $\text{Prob}[X_u=1]$ to 0.5 as shown in Fig. \ref{fig:hep_5_1}, A-Greedy becomes 320\% better than Greedy. As shown in Figs. \ref{fig:pl_1} \ref{fig:pl_2} and \ref{fig:wiki_1}, for PL and Wiki network, we have the similar result. For example, for the PL network under $\mathscr{F}^1$ with $\text{Prob}[X_u=1]=0.5$, one seed node results about 2.5 active nodes under A-Greedy while in average 1.67 nodes can be activated by a single seed node under Greedy. Another important observation is that the curves generated by Greedy become less stable in the DIC model, which implies that to reach the same level of accuracy Greedy requires more number of simulations than A-Greedy does.

Now let us discuss the performance of the proposed heuristic seeding strategy H-Greedy. Fig. \ref{fig:pathtest} shows the distribution of $E[H(v)]$ drew from the dataset by simulation. In Fig. \ref{fig:path_1}, 90 \% of the nodes cannot activate more than 2 nodes, while in Figs. \ref{fig:path_2} and \ref{fig:path_3}, we can see that there is a significant gap between the strength of influential nodes and that of other nodes. For example, as shown in Fig. \ref{fig:path_2}, 24 percent of the nodes in Wiki can activate more than 1600 nodes while 82 percent of them can hardly activate more than 50 nodes. For PL dataset in the same setting, about 30 percent of the nodes could bring 780 active nodes while 68 percent of them only results less than 100 active nodes. Admitting that the difference of $E[H(v)]$ between two nodes would decrease along with the seeding process due to the submodularity, the nodes with small $E[H(v)]$ are not likely to be a seed node as the gap is too large and we only have a small budget compared to the population of users. Thus, 1-sigma control on $E[H(v)]$ is a safe bound such that we will not miss any influential nodes. As shown in Fig. \ref{fig:a-greedy}, under all the circumstances the performance of H-Greedy is almost the same as that of A-Greedy. This is because in those settings H-Greedy can hardly eliminate any nodes as the distributions of $E[H(v)]$ are like Fig. \ref{fig:path_1}. Thus, H-Greedy is identical to A-Greedy in those cases. However, for the cases where the distribution of $E[H(v)]$ has a pattern like Figs \ref{fig:path_2} or \ref{fig:path_3}, H-Greedy would be an effective and efficient strategy. In these cases, H-Greedy could rule out more than a half of the nodes from the candidate seed nodes and thus more than 20\% time consumed in the seeding process could be saved as shown in Fig. \ref{fig:time}. Furthermore, H-Greedy performs slightly worse than A-Greedy but still better than Greedy, as shown in Fig. \ref{fig:pl_1_4_h} and \ref{fig:pl_5_5_h}.





\begin{table}[t]
\centering
\small
{\begin{tabular}{ p{4.8cm}   p{1.42cm}  p{1.4cm} }

\hline
Parameter Setting & H-Greedy \newline (ms)& A-Greedy \newline (ms)\\ 
\hline 
$\mathscr{F}^2$ \& Prob$\text{[}X_u=1\text{]}=1$ on PL & 14977& 51485 \\
$\mathscr{F}^2$ \& Prob$\text{[}X_u=1\text{]}=1$ on Wiki & 87412& 268499 \\ 
$\mathscr{F}^3$ \& Prob$\text{[}X_u=1\text{]}=1$ on PL & 981& 11931 \\
$\mathscr{F}^3$ \& Prob$\text{[}X_u=1\text{]}=1$ on Wiki & 31247& 44625 \\
\hline
\end{tabular}}
\caption{Scalability of H-Greedy. The four cases are shown in the first column. The second and third column shows the average time consumed in selecting one seed node under H-Greedy and A-Greedy.}
\label{table:time}
\end{table}

\section{Conclusion and Future work}
\label{sec:conclusion}
In this paper we have considered the problem that how to maximize the spread of influence in dynamic social networks. The proposed DIC model is able to capture the dynamic aspects of a real social network and the uncertainness of the diffusion process. In the DIC model, a certain node can be seeded for more than one time and the propagation probability between two users varies following a certain distribution. Based on the DIC model, we formulate the adaptive seeding strategies by introducing the concept of seeding pattern. The pattern $A^{*}$ constructed in Sec. \ref{sec:model} shows the optimal method to determining how many budgets shall we utilize in each seeding step. Combining the optimal pattern with the natural hill-climbing algorithm, we present the A-Greedy seeding strategy and show that A-Greedy has a performance ratio of $(1-1/e)$. By the observation that the influential nodes are much more powerful than other nodes in a social network, we further design an simple heuristic adaptive seeding strategy H-Greedy based on A-Greedy. The experimental results herein demonstrate the superiority of the adaptive seeding strategies to prior approaches.

The future work of this topic consists of several aspects. As we can see, H-Greedy is a simple heuristic strategy and it is not effective for all the settings of DIC model. Thus, we plan to design better heuristic adaptive seeding strategies that are able to deal with general social networks. We note that the technique in \cite{chen2010scalable} is possibly applicable to the adaptive seeding framework and we leave this part as future work. Another aspect of the future work is to design adaptive seeding strategies which are able to meet the round limit. In real applications, we may only care about the spread influence within a certain number of rounds. In this case, the analysis of the adaptive seeding strategies becomes intricate. On the one hand as shown by pattern $A^{*}$ we try to utilize the budgets as late as possible in order to obtain more information while on the other hand delaying a seeding step leads us to lost a diffusion round as we have round limit. One can easily check that with a round limit our objective function is not submodular anymore, which renders it more hard to find a greedy algorithm with a provable performance guarantee.

%




\ifCLASSOPTIONcaptionsoff
  \newpage
\fi




\bibliographystyle{IEEEtran}
\bibliography{sigproc}
%



%

\begin{IEEEbiography}[{\includegraphics[width=1in,clip,keepaspectratio]{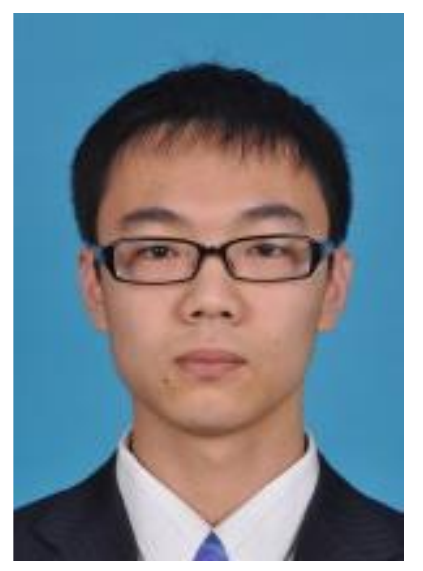}}]{Guangmo Tong} is a Ph.D candidate in the Department of Computer Science at the University of Texas at Dallas. He received his BS degree in Mathematics and Applied Mathematics from Beijing Institute of Technology in July 2013. His research interests include real-time and embedded systems and social networks. He is a student member of the IEEE.
\end{IEEEbiography}

\begin{IEEEbiography}[{\includegraphics[width=1in,clip,keepaspectratio]{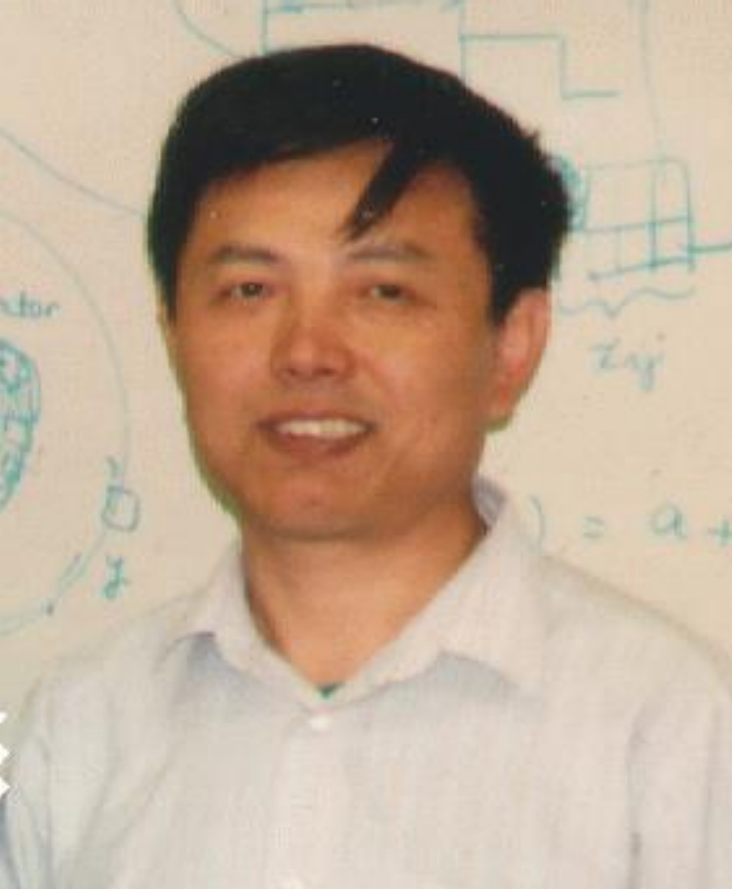}}]{Ding-Zhu Du} received the M.S. degree from the Chinese Academy of Sciences in 1982 and the Ph.D. degree from the University of California at Santa Barbara in 1985, under the supervision of Professor Ronald V. Book. Before settling at the University of Texas at Dallas, he worked as a professor in the Department of Computer Science and Engineering, University of Minnesota. He also worked at the Mathematical Sciences Research Institute, Berkeley, for one year, in the Department of Mathematics, Massachusetts Institute of Technology, for one year, and in the Department of Computer Science, Princeton University, for one and a half years. He is the editor-in-chief of the Journal of Combinatorial Optimization and is also on the editorial boards for several other journals. Forty Ph.D. students have graduated under his supervision. He is a member of the IEEE
\end{IEEEbiography}





\end{document}